\documentclass[11pt]{article}

\RequirePackage{amsmath,amsthm,amsfonts,dsfont}
\RequirePackage{amssymb,mathrsfs}

\RequirePackage{fullpage}
\RequirePackage[defaultlines=3,all]{nowidow}  \RequirePackage{lmodern} \usepackage[utf8]{inputenc}
\usepackage[T1]{fontenc}
\RequirePackage{microtype} 

\RequirePackage{enumitem}
\setitemize[1]{nosep}

\RequirePackage{mathtools}
\mathtoolsset{centercolon}  \RequirePackage{bm} \RequirePackage[font=small,labelfont=bf]{caption} \RequirePackage{booktabs} \RequirePackage{graphicx}
\RequirePackage{mleftright}
\mleftright

\newcommand{\mytableofcontents}{\thispagestyle{empty}
    \newpage
    \setcounter{tocdepth}{2}
    \tableofcontents
    \newpage
    \pagenumbering{arabic}}

\RequirePackage{xcolor}
\RequirePackage[hyphens]{url}
\RequirePackage{hyperref}
\hypersetup{
    breaklinks=true,
    colorlinks,
    linkcolor={blue!65!black},
    citecolor={blue!65!black},
    urlcolor={red!65!black},
    filecolor={red!65!black},
}

\RequirePackage[capitalize,nameinlink,noabbrev,nosort]{cleveref} \usepackage[hyperpageref]{backref} 

\newtheorem{theorem}{Theorem}[section]
\usepackage{thmtools, thm-restate}

\newtheorem{lemma}[theorem]{Lemma}

\newtheorem{corollary}[theorem]{Corollary}
\newtheorem{definition}[theorem]{Definition}

\newcommand{\N}{\ensuremath{\mathbb{N}}}

\newcommand{\R}{\ensuremath{\mathbb{R}}}

\newcommand{\Z}{\ensuremath{\mathbb{Z}}}

\newcommand{\eps}{\varepsilon} 

\DeclareMathOperator{\poly}{poly}

\DeclareMathOperator{\polylog}{polylog}

\DeclareMathOperator{\vol}{vol}

\DeclareMathOperator*{\argmin}{argmin}

\DeclareMathOperator*{\expect}{\mathbb{E}}

\RequirePackage{derivative}

\renewcommand{\epsilon}{\eps} 
\renewcommand{\vec}[1]{\bm{#1}}

\renewcommand{\tilde}{\widetilde}

\DeclarePairedDelimiter\floor{\lfloor}{\rfloor}
\DeclarePairedDelimiter\ceil{\lceil}{\rceil}

\usepackage{algorithm}
\usepackage{algpseudocodex}
\usepackage{mdframed}
\usepackage{framed}
\usepackage{bbm}

\usepackage[most]{tcolorbox}

\newtcolorbox{algbox}[1]{
    enhanced,
    breakable,
    colback=gray!4,
    colframe=black!60,
    title=\textbf{#1},
    fonttitle=\bfseries,
    boxrule=0.6pt,
    arc=2pt,
    left=8pt,
    right=8pt,
    top=6pt,
    bottom=6pt
}

\usepackage{subcaption}

\crefname{section}{\S\!}{\S\S\!}
\Crefname{section}{\S\!}{\S\S\!}

\def\tO{\tilde{O}}
\def\cP{\mathcal{P}}
\def\cX{\mathcal{X}}
\def\cM{\mathcal{M}}
\def\cO{\mathcal{O}}

\def\Round{\mathfrak{R}}

\newcommand{\repr}{r}

\newcommand{\Paren}[1]{\left(#1\right)}

\newcommand{\sens}{\mathrm{sens}}

\newcommand{\scr}{\mathrm{scr}}

\newcommand{\fr}[1]{\left\{#1\right\}}

\newcommand{\one}{\mathbf 1}

\newcommand{\etanet}{\mathcal{N}_{\eta}}

\newcommand{\calE}{\mathcal{E}}

\newcommand{\calB}{\mathcal{B}}

\newcommand{\calP}{\mathcal{P}}

\newcommand{\msspset}{\mathbb{S}}

\newcommand{\reals}{\mathbb{R}}

\newcommand{\integers}{\mathbb{Z}}

\newcommand{\torus}{\mathbb{T}}

\newcommand{\bbB}{\mathbb{B}}

\newcommand{\ExpPart}{\mathsf{ExpPart}}

\newtheorem{observation}{Observation}

\allowdisplaybreaks

\usepackage{tikz}
\usetikzlibrary{arrows.meta,positioning,fit,calc,decorations.pathreplacing,matrix,external}

\title{
Overcoming the Randomness-Utility Trade-off in\\
Answering Differentially Private Linear Queries
}

\author{
  \begin{tabular}{cc}
    \begin{tabular}[t]{c}
      Surendra Ghentiyala\thanks{Work done during an internship at Google Research.}\\
      Cornell University\\
      \texttt{sg974@cornell.edu}
    \end{tabular}
    &
    \begin{tabular}[t]{c}
      Pritish Kamath\\
      Google Research\\
      \texttt{pritish@alum.mit.edu}
    \end{tabular} \\[1.5cm] \begin{tabular}[t]{c}
      Ravi Kumar\\
      Google Research\\
      \texttt{ravi.k53@gmail.com}
    \end{tabular}
    &
    \begin{tabular}[t]{c}
      Pasin Manurangsi\\
      Google Research\\
      \texttt{pasin@google.com}
    \end{tabular}
  \end{tabular}
}

\date{\today}
\begin{document}

\maketitle

\begin{abstract}
We study the question of answering linear queries with differential privacy using \emph{few} (expected) random bits. We provide a randomness-efficient analog of the $\| \cdot \|_K$-norm mechanism of Hardt and Talwar \cite{HardtT10}. For the $\ell_\infty$-error, our algorithm can answer $d$ linear queries with $O(d / \eps)$ error using $O(\log d)$ random bits, improving upon algorithms of~\cite{canonne2025randomness,ghentiyala2026efficient}; this is optimal when $\eps \le 1/d$.
We also provide a computationally efficient version of our algorithm, albeit with an $O(\log d)$ multiplicative increase in the error.

\end{abstract}

\thispagestyle{empty}
\newpage
\addtocounter{page}{-1}

\mytableofcontents
\newpage

\section{Introduction}

As differential privacy (DP)~\cite{dwork2006calibrating} has become a gold standard in privacy-preserving algorithms, it has seen large-scale adoption in various real-world systems, both in industry and government sectors. This adoption has revealed several 
practical challenges in applying DP.  One such challenge
is \emph{randomness complexity}---the (expected) number of independent, uniformly random bits required to implement these inherently randomized algorithms.  In the US Census deployment of DP, for example, an estimated 7.2 $\times 10^{14}$ random bits were used~\cite{GarfinkelL20}. As such, the randomness complexity has become another important resource constraint---similar to space/time complexity and utility---that can determine the practical applicability of DP algorithms. 

Motivated by these practical considerations, Canonne et al.~\cite{canonne2025randomness} initiate the theoretical study of the randomness complexity of DP algorithms. They focus on arguably the simplest and most fundamental problem: Answering $d$ linear queries. In this setting, one can simply apply the classic $\eps$-DP (discrete) Laplace mechanism~\cite{dwork2006calibrating,GhoshRS12} to achieve an $\ell_\infty$-error of $O(d \log d / \eps)$. However, this requires generating independent noise for every one of the $d$ queries, resulting in randomness complexity of at least $\Omega(d)$. In a surprising result, Canonne et al.~\cite{canonne2025randomness} show that the randomness complexity can be significantly reduced. In particular, they give an algorithm\footnote{They also provide an algorithm that satisfies \emph{approximate}-DP. Nevertheless, our \emph{pure}-DP algorithm already improves on it in the low-randomness-complexity regime. See~\Cref{sec:open} for a more detailed discussion.} that uses only $O(\log d)$ random bits in expectation and achieves an $\ell_\infty$-error of $\tO(d^2 / \eps)$. In their algorithm, it is also possible to trade off more randomness complexity with lower error. Nevertheless, at the other extreme of $O(d \log d / \eps)$ error, matching the Laplace mechanism, their algorithm's randomness complexity remains $\Omega(d)$, similar to the Laplace mechanism. A subsequent work of Ghentiyala~\cite{ghentiyala2026efficient} simplified and improved the efficiency of their algorithm.
More recently, Kalinin and Pagh \cite{kalinin2026dithered} explored the randomness complexity (and practical efficiency) of variants of the Gaussian and Laplace mechanism in a model that allows access to both public and private randomness. Still, the overall randomness complexity vs utility tradeoff has remained unchanged from  \cite{canonne2025randomness} up to logarithmic factors. This brings us to the central question we study:
\begin{quote}\centering\slshape
What is the inherent tradeoff between the randomness complexity \\ and utility  for answering DP linear queries?
\end{quote}

\subsection{Our Results}
Surprisingly, we show that there is \emph{no} tradeoff at all: We can achieve both (nearly) optimal utility and randomness complexity simultaneously, as stated below.

\begin{theorem}[Informal; See \cref{cor:lp-mech}]
There is an $\eps$-DP mechanism that can answer $d$ linear queries with expected $\ell_\infty$-error $O(d / \eps)$ and randomness complexity $O(\log d)$.
\end{theorem}
As we show later, the above theorem also holds for a broad family of norms, but for the discussion below, we focus on the $\ell_\infty$-norm.  Note that in the $\ell_\infty$-norm, the error bound is even better than the standard Laplace mechanism~\cite{dwork2006calibrating,GhoshRS12} by an $O(\log d)$ factor and is in fact optimal~\cite{HardtT10,SteinkeU16}.
As for the randomness complexity, \cite{canonne2025randomness} prove that $\Omega(\log d)$ is necessary in the high-privacy regime where $\eps \leq 1/d$. It remains an interesting open question to extend this lower bound to larger $\eps$.

Though the above mechanism is \emph{not} efficient, it demonstrates that $O(\log d)$ bits of randomness (in expectation) is sufficient to answer linear queries with optimal expected error.
Nevertheless, we also give a polynomial time mechanism for this task, but with an error bound that is an $O(\log d)$ factor larger than the above bound (hence matching that of the Laplace mechanism).

\begin{theorem}[Informal; See \cref{thm:efficient_main}] Let $\varepsilon < \log d$. There is an efficient $\eps$-DP mechanism that can answer $d$ linear queries with expected $\ell_\infty$-error $O((d \log d) / \eps)$ and randomness complexity $O(\log d)$.
\end{theorem}

\subsection{Technical Overview}

For a dataset $X \in \cX^*$, let $f: \cX^* \to \R^d$ denote the function of interest. Throughout this overview, we assume that its $\ell_\infty$-sensitivity is bounded by one. In the case of linear queries, we adhere to the standard DP convention that $f(x) = \sum \phi_i(x_i)$ where $\phi_i$ has co-domain $[0, 1]$. Therefore, the assumption that $f$ has $\ell_{\infty}$-sensitivity bounded by one generalizes the $d$ linear queries setting, as we can view $f(X)_i$ as the answer to the $i$th linear query.

\paragraph{Secluded Partition.}
To describe our algorithm and that of \cite{canonne2025randomness}, we recall the crucial notion of \emph{secluded partition} and its construction from\footnote{See also \cite{WDPRV24}, which is a conference version of these works.} Woude et al.~\cite{WDPRV22,WDPRV23}.

\begin{definition}[Secluded Partition]
\label{def:secluded-partition}
For $\rho > 0$, a partition $\cP$ of $\R^d$ is called a \emph{$(k, \gamma; \rho)$-secluded partition} with respect to the norm $\| \cdot \|$ if there exists a function $\repr : \cP \to \R^d$ such that $\repr(P) \in P$ for all $P \in \cP$ such that the following hold:
\begin{itemize}
\item For all $P \in \cP$ and for all $x \in P$, it holds that $\|\repr(P) - x\| \le \rho$.
\item Any $\| \cdot \|$ ball of radius $\gamma \cdot \rho$ intersects at most $k$ cells in $\cP$.
\end{itemize}
Sometimes, we omit the parameter $\rho$ to mean that it is a $(k, \gamma; \rho)$ secluded partition for some $\rho$.\footnote{Note that a $(k, \gamma; \rho)$-secluded partition can be scaled to obtain a $(k, \gamma; \rho')$-secluded partition for all $\rho' > 0$.}
\end{definition}

We refer to $\repr(P)$ as the \emph{representative} of $P$, and the map $\Round : x \mapsto r(P)$ where $x \in P \in \cP$ as the {\em rounding function}.
Given a secluded partition $\cP$, the algorithm in \cite{canonne2025randomness} can be loosely described as follows: 
\begin{itemize}
\item Add Laplace noise to $f(X)$ to obtain $Y$.
\item Find the partition cell $P \in \cP$ to which $Y$ belongs.
\item Output the representative $r(P)$ of cell $P$.
\end{itemize}
\medskip
The scale parameter $\rho$ is selected so that the Laplace noise added is of norm at most $\gamma \cdot \rho$ with high probability.
The crux of the analysis from \cite{canonne2025randomness} is that, due to the second property of the secluded partition, with high probability the output of the mechanism will be one of these $k$ partition cells. Via a slightly more careful analysis, one can show that this implies that the entropy of the output is also $O(\log k)$. A standard argument by Knuth and Yao~\cite{knuth1976complexity} then implies that this distribution can be sampled using $O(\log k)$ independent bits of randomness (in expectation).

In terms of the high-probability error bound, since the norm of the Laplace noise is at most $\gamma \cdot \rho$ and rounding to the representative adds an error of at most $\rho$,
the triangle inequality ensures that the error is $O(\rho)$.

\paragraph{Our Algorithm.}

Instead of a two-stage analysis of Laplace noise and rounding, our main idea is to forego the noise, and simply use the exponential mechanism~\cite{McSherryT07} on the representatives of the secluded partition, using a score of the distance of $f(X)$ from the corresponding cells.
More formally, let the distance from a point $z \in \R^d$ to a partition cell $P \in \cP$ be defined as
$\Delta_P(z) := \inf_{y \in P} \| z- y\|_\infty$ (we only consider $\| \cdot \|_\infty$ in this overview, though our mechanism works for a large class of norms; see \Cref{subsec:k-norm}).
Crucially, we define the distance to a cell $P$ to be the minimum distance to {\em any} point in $P$ rather than the distance to the representative of $P$; this will prove important in our analysis.

It is easy to see that the sensitivity of $\Delta_P(f(\cdot))$ is bounded by the sensitivity of $f$ (in the given norm). As such, the exponential mechanism is $\eps$-DP as desired. Thus, most of our technical contribution is in analyzing the randomness complexity of this mechanism.

\paragraph{Randomness Complexity.}
Once again, our proof of the randomness complexity essentially boils down to showing that when the scaling factor $\rho$ of a secluded partition is selected appropriately, the output is, with high probability, from a small number of partitions.

To explain the proof and parameter selection in more detail, we need to be more specific about the parameters of secluded partitions. First, let us recall the construction from \cite{WDPRV22,WDPRV23}\footnote{Note that the bounds in \cite{WDPRV22,WDPRV23} are more refined. However, since we are only interested in asymptotic results, we consider this simplified version.}:
\begin{theorem} [\cite{WDPRV23}] \label{thm:wdprv-secluded-partition}
For every $\ell \in [d]$, there is an $((\ell + 1)^{O(d/\ell)}, 1/2\ell)$-secluded-partition.
\end{theorem}

Crucially however, the choice of the partition in the above construction depends on the value of $\ell$. Our first contribution is to show that there is, in fact, a \emph{single} partition that is 
simultaneously secluded with similar parameters \emph{for all} $\ell$.
We refer to this as a \emph{multi-scale secluded-partition}.
\begin{restatable}[Multi-Scale Secluded Partition]{definition}{mssp}
\label{conj:multiscale}
Let $\rho \in \R_{> 0}$, $d \in \integers_{> 0}$ and $k_d: [d] \to \integers_{>0}$ be a function.
$\cP$ is a \emph{$(k_d, \rho)$-multi-scale secluded partition} ($(k_d, \rho)$-MSSP) with respect to the norm $\| \cdot \|$ if for all $\ell \in [d]$, $\cP$ is a $(k_d(\ell), 1/\ell; \rho)$-secluded partition with respect to the norm $\| \cdot \|$.
\end{restatable}
\noindent We omit $\rho$ in $(k_d, \rho)$-MSSP to mean ``for some $\rho$'' (again, one can modify the $\rho$ easily via scaling). We use the probabilistic method to prove the existence of such a MSSP.
\begin{theorem}[Informal; See \cref{thm:k_norm_part}]
\label{lem:mssp}
For any $d \in \N$, there is a $k_d$-MSSP for the $\ell_\infty$-norm, where
$k_d(\ell) = d^{O(1)} \cdot 2^{O(d/\ell)}$.
\end{theorem}
We note that, interestingly, the parameters we obtain in the above theorem are \emph{better} than those of the (single-scale) secluded-partition from \cite{WDPRV23} in certain regimes, e.g.,
$\Theta(1) \leq \ell \leq \Theta(d/\log d)$; see \Cref{thm:wdprv-secluded-partition}.   This might be of independent interest.

\paragraph{Proof Overview.}
We are now ready to discuss the core argument in the proof. We set $\rho = \Theta\Paren{d/\eps}$ where ``$\Theta$'' contains a sufficiently large constant. Then, we claim that the algorithm will, with high probability, output only a partition cell that is within distance $\frac{\rho \log d}{d}$ of $f(X)$. Note that this implies that the error is at most $\rho + \frac{\rho \log d}{d} = O\Paren{d/\eps}$. Moreover, since there are only $d^{O(1)}$ such partition cells, it is not hard to further show
that the entropy of the output is at most $O(\log d)$, which yields the desired randomness complexity.

The high-level intuition for the claim is as follows. First, it is not hard to show that faraway partition cells, i.e., those with $\Delta_P(f(X)) > \rho$, are selected with very low total probability.
Among the remaining partition cells, we can group them into ``rings'' based on distances; namely, for each $\ell \in [d/\log d]$, we can consider the group of all cells $P$ where $\Delta_P(f(X)) \in \left[\frac{\rho}{\ell+1}, \frac{\rho}{\ell}\right)$. Since $\cP$ is an MSSP, by \Cref{lem:mssp}, there are at most $d^{O(1)}2^{O(d/\ell)} \leq \exp\Paren{O(d/\ell)}$ such cells. However, the exponential mechanism assigns probability at most $\exp\Paren{-\frac{\eps}{2} \cdot \frac{\rho}{\ell + 1}} = \exp\Paren{-\Theta\Paren{d/\ell}}$ to each such cell. Thus, by picking the constant in big-$\Theta$ to be sufficiently large, the probability of selecting any such cell is small. Finally, a union bound over all these rings concludes the proof. 

\paragraph{\boldmath Extensions to Other $K$-Norms.} Finally, as an astute reader might have already noticed, our mechanism does not use any specific property of the $\ell_\infty$-norm, apart from the existence of a multi-scale secluded-partition. It turns out that, via a simple probabilistic argument, a multi-scale secluded-partition exists very generally across a large class of $K$-norms for ``somewhat nice'' $K$. As a result, it is simple to extend our mechanism to match the error of the so-called $K$-norm mechanism by Hardt and Talwar~\cite{HardtT10}, where $K$ is the {\em sensitivity polytope} (\Cref{def:sensitivity-polytope}) of the function $f$. In particular, for $\ell_p$-norms, we can still get an $\ell_p$-error of $O(d/\eps)$ using $O(\log d)$ bits of randomness.

\paragraph{Computationally Efficient Mechanism.}
The probabilistic construction of an MSSP does not immediately yield an efficient mechanism. However, for the special case of the $\ell_\infty$-norm, we construct an explicit MSSP, and show that the algorithm
using it can be implemented efficiently.
The MSSP, however, has slightly worse parameters, leading to an error of $O((d \log d)/\eps)$ with randomness complexity of $O(\log d)$.
The MSSP itself can be viewed as a modification of the partition by Hoza and Klivans~\cite{HozaK18}. The parameters we obtain for this MSSP essentially match those of \cite{WDPRV22,WDPRV23} (\Cref{thm:wdprv-secluded-partition}), which are worse than the probabilistic construction. This results in the $O(\log d)$ factor increase in the error. Closing this gap is an interesting open problem.
 \section{Preliminaries}
\label{sec:prelim}

\subsection{Convex Geometry}

Let $\calB_p^d(r)$ denote the
$\ell_p$-ball of radius $r$ in $\R^d$.  
For any absorbing disk\footnote{A set $K \subseteq \R^d$ is an \emph{absorbing disk} if (i) it is convex, (ii) it is symmetric, i.e., $K = -K$, and (iii) for all $z \in \R^d$, there exists $t \in \R_{\ge 0}$ such that $z \in tK$.} $K$ that does not contain a nontrivial subspace of $\reals^d$, define the \emph{$\| \cdot \|_K$-norm} as
\begin{align}
\label{eq:knorm}
\| z \|_K & := \inf \left\{ t \geq 0 : z \in tK \right\}. \end{align}
Throughout this work, we assume that $K$ is an absorbing disk not containing a nontrivial subspace of $\reals^d$ and hence satisfies the conditions needed to make $\| \cdot \|_K$ a norm. Note that the notion of $\| \cdot \|_K$-norm subsumes all $\ell_p$-norms.

\begin{definition}[$\delta$-net]
For a set $K \subset \reals^d$ and norm $\| \cdot \|$ on $\reals^d$, we say that a finite set $A \subseteq K$ is a \emph{$\delta$-net} for $K$ with respect to $\| \cdot \|$ if for all $x \in K$, there exists $y \in A$ such that $\| x - y\| \leq \delta$.
\end{definition}

\begin{lemma}[{\cite[Lemma 4.16 rephrased]{pisier1999volume}}]
    \label{lem: covering_num_bound}
    Let $\| \cdot \|_1, \| \cdot \|_2$ be two arbitrary norms on $\reals^d$ with respective unit balls $B_1, B_2$ where $B_2 \subseteq B_1$. Let $\delta > 0$. There is a finite set $A \subset B_1$ that is a $\delta$-net for $B_1$ with respect to $\| \cdot \|_2$ with
    \[ |A| \leq \left( 1 + \frac{2}{\delta} \right)^d \cdot \frac{\vol(B_1)}{\vol(B_2)} .\]
\end{lemma}

\begin{theorem}[John's Theorem~\cite{john1948extremum};  {\cite[Theorem 3.13 rephrased]{tao2006additive}}]
    \label{thm:john}
    Let $K$ be a symmetric convex body in $\reals^d$ (namely $K = -K$). Then there exists an invertible linear transformation $T: \reals^d \rightarrow \reals^d$ such that 
    \[ \calB_2^d(1) \subseteq T K \subseteq \calB_2^d(\sqrt{d}), \qquad \text{where } TK := \{ Tz \mid z \in K\}.\]
\end{theorem}

\subsection{Differential Privacy (DP)}

We recall the definition of (pure-) DP~\cite{dwork2006calibrating} below. The definition relies on a symmetric notion of {\em neighboring datasets}, denoted $X \sim X'$; some typical neighborhood notions are adding/removing one element, or replacing one element to obtain dataset $X'$ from $X$ (or vice versa). A {\em mechanism} $\cM : \cX^* \to \Delta_{\cO}$ maps datasets $X \in \cX^*$ to probability measures over the output space $\cO$.

\begin{definition}[Differential Privacy]
A mechanism
$\cM: \cX^* \to \Delta_{\cO}$ is \emph{$\eps$-differentially private ($\eps$-DP)} if, for any pair $X \sim X'$ of neighboring datasets and any measurable set $E \subseteq \cO$, it holds that $\Pr[\cM(X) \in E] \leq e^{\eps} \cdot \Pr[\cM(X') \in E]$.

When $\cO$ is countable, this is equivalent to requiring that, for all $o \in \cO$, $\Pr[\cM(X) = o] \leq e^{\eps} \cdot \Pr[\cM(X') = o]$.
\end{definition}

\noindent For the purposes of our results, the exact neighborhood notion is irrelevant, as long as the function $f$ has low \emph{sensitivity} for neighboring datasets.

\begin{definition}[Sensitivity]
Let $f: \cX^* \to \R^d$ be a function, and $\|\cdot\|$ be any norm over $\R^d$. The sensitivity of $f$ with respect to $\|\cdot\|$ is defined as
$\sens_{\|\cdot\|}(f) = \sup_{X \sim X'} \|f(X) - f(X')\|.$
\end{definition}

For the special case of $d = 1$, we assume throughout that the norm is the absolute value and will henceforth omit it from the notation.
It is also useful to define the notion of a \emph{sensitivity polytope}, as follows.

\begin{definition}[Sensitivity Polytope]\label{def:sensitivity-polytope}
We say that $K \subseteq \R^d$ is a \emph{sensitivity polytope} of $f$ if $f(X) - f(X') \in K$ for all neighboring datasets $X \sim X'$.
\end{definition}
Note that in the above definition, a sensitivity polytope for $f$ might not be unique. This is somewhat more convenient for our exposition below (instead of defining the sensitivity polytope to be the unique smallest such polytope).

The sensitivity of $f$ and the sensitive polytope of $f$ relate to each other as follows:
\begin{observation}
Let $K$ be a sensitivity polytope of $f$. Then, we have $\sens_{\|\cdot\|_K}(f) \leq 1$.
\end{observation}
\begin{proof}
For any neighboring $X \sim X'$, we have $f(X) - f(X') \in K$. Using the definition of  $\| \cdot \|_K$-norm (see \eqref{eq:knorm}), it follows that $\|f(X) - f(X')\|_K \leq 1$.
\end{proof}
The \emph{exponential mechanism}  \cite{McSherryT07} is parameterized by the output set $\cO$ and a
\emph{scoring function} $\{\scr_o\}_{o \in \cO}$, where $\scr_o: \cX^* \to \R$. The algorithm simply outputs each $o \in \cO$ with probability proportional to $\exp\Paren{-\frac{\eps}{2} \cdot \scr_o(X)}$. The privacy guarantee of the mechanism is given below:
\begin{theorem}[Exponential Mechanism, \cite{McSherryT07}] \label{thm:exm-dp}
If $\scr_o$ has sensitivity at most 1 for all $o \in \cO$, then the exponential mechanism is $\eps$-DP.
\end{theorem}
Adopting the conventions of \cite{canonne2025randomness}, we now define the output entropy and randomness complexity of a mechanism. 

\begin{definition}[Output Entropy]
    \label{def:output_entropy}
    The \emph{output entropy} of a mechanism $\cM: \cX^* \to \Delta_{\cO}$ is defined as $\sup_{X \in \cX^*} H(\mathcal{M}(X))$, where $H(\cdot)$ is the Shannon entropy.
\end{definition}

\begin{definition}[Randomness Complexity]
    \label{def:rand_complexity}
    The \emph{randomness complexity} of a mechanism $\cM: \cX^* \to \Delta_{\cO}$ is defined as the supremum over $X \in \cX^*$ of the expected number of random bits used by $\cM(X)$.
\end{definition}

Output entropy is a purely information-theoretic notion while randomness complexity is also a computational notion. We have intentionally left out the model of computation in \cref{def:rand_complexity}. When we design our efficient mechanism in \cref{sec:efficient}, we will work with the standard notion of a probabilistic Turing machine running in polynomial time. However, in all other instances, we will assume that the model of computation for which we are measuring randomness complexity is as general as possible (e.g., an infinite decision tree reading the input and some number of random bits). This will allow us to use the DDG-tree sampler of Knuth and Yao \cite{knuth1976complexity} to smoothly transition between output entropy and randomness complexity.

\begin{theorem}[paraphrased \cite{knuth1976complexity}]
    \label{thm:KY}
    Let $\vec{p} = (p_1, \dots, p_n)$ be a discrete probability distribution (where $n$ may be infinite).  The average number of random bits used by an optimum DDG-tree algorithm for $\vec{p}$ is at most $H(\vec{p}) + 2$ where $H(\cdot)$ denotes the binary entropy function.
\end{theorem}

For the algorithm described in \Cref{sec:efficient}, it is essential to provide an \emph{efficient} method for sampling from the distribution rather than merely establishing a bound on the entropy. For this, we will need to generate a $\mathrm{Ber}(p)$ random variable as a subroutine. This can be done efficiently and with low randomness, since it is
a special case of the Knuth--Yao sampler \cite{knuth1976complexity} where the DDG-tree construction and execution are polynomial time (a formal proof of a more general version of \cref{lem:knuth-yao} can be found in the appendix of \cite{ghentiyala2026efficient}).

\begin{lemma}
    \label{lem:knuth-yao}
    For a $t$-bit value $p$, we can sample from $\mathrm{Ber}(p)$ in $\poly(t)$ time and using $O(1)$ bits of randomness in expectation.
\end{lemma}
 \section{Multi-Scale Secluded Partitions}
In this section, we will consider multi-scale secluded partitions (MSSPs) solely as a geometric object and show their existence, an explicit construction, and useful properties.  We recall the definition.
\mssp*
MSSPs generalize secluded partitions (see \cref{def:secluded-partition} and \cref{thm:wdprv-secluded-partition}) by demanding a stronger order of quantifiers. Rather than asking that for every $\ell$, there exist a $(k_d(\ell), 1/\ell; \rho)$-secluded partition $\calP_\ell$, we require a single $\calP$ that is a $(k_d(\ell), 1/\ell; \rho)$-secluded partition for all $\ell \in [d]$.

\subsection{\boldmath Existence for any \texorpdfstring{$\| \cdot \|_K$}{K}-norm}
\label{subsec:k-norm}

We will first show that for $K$ that are sandwiched between inner and outer $\ell_2$-balls of appropriate sizes, there exists an MSSP for the $\| \cdot \|_K$ norm. We will then generalize this to any $K$ by applying the proper transformation.
\begin{lemma}
    \label{lem:john_K}
    Let $K \subset \reals^d$ be any set such that $\| \cdot \|_K$ is a norm and $\calB^d_2(1) \subseteq K \subseteq \calB_2^d(\sqrt{d})$. There exists a $k_d$-MSSP for $\| \cdot \|_{K}$, where $k_d(\ell) = O(d \log d \cdot e^{d/ \ell})$.
\end{lemma}
\begin{proof}
    Let $\torus = \reals^d/(L \integers^d)$ be the flat torus with side length $L = 20 \sqrt{d}$.  Let the
    distance in the torus be given by
    $d_{\torus}(x, y) := \min_{z \in L \integers^d} \| x - y + z \|_K$.

    Let $\eta = 1/(4d)$. We construct an $\eta$-net $\etanet$ on the torus $\torus$, in which the elements of the net will be our \emph{anchor points}. In fact, it suffices to construct the $\eta$-net on $[-L/2, L/2]^d$ since distances on the flat torus cannot be any greater than on $[-L/2, L/2]^d$. By \cref{lem: covering_num_bound},
    \begin{align*}
    |\etanet| 
    &\leq~ \left( 1 + \frac{2}{1/(4d)} \right)^d \cdot \frac{\vol([-L/2, L/2]^d)}{\vol(K)}\\
    &\leq~ (9d)^d \cdot \frac{\vol([-L/2, L/2]^d)}{\vol(\calB_2^d(1))}\\
    &\leq~ (9d)^d (20 \sqrt{d})^d \left( d ! \right)\\
    &=~ 2^{O(d \log d)}.
    \end{align*}
Let $\lambda = C \cdot d \log(d+1)$ for a sufficiently large constant $C$. We choose 
    \[ s = \left\lceil \lambda \cdot \frac{\vol(\torus)}{\vol(K)} \right\rceil, \]
    independent points uniformly from $\torus$ to create a subset $S$.

    Our secluded partition will correspond to the Voronoi cells of each point in $S$ with respect to $d_{\torus}(\cdot, \cdot)$. Formally, the partition is defined as $\cP_{S} := \{ P_w \mid w \in S \}$ where $P_w := \{ x \in \R^d \mid w = \argmin_{w' \in \msspset} d_{\torus}(x, w') \}$, and the representative function is $r(P_w) := w$ for all $w \in S$.
We now show that both the covering and secludedness properties hold with high probability.
    \paragraph{Covering.} We first show that with high probability, every point $x \in \torus$ is within distance 1 of some point in $S$. To do so, we will show that with high probability each anchor point is within distance $1-\eta$ of some point in $S$. Note that this is sufficient as it would imply that for any $x \in \torus$, $x$ is within distance $\eta$ of an anchor point, which is itself within distance $1-\eta$ of $S$, so $x$ is within distance $1$ of $S$ by the triangle inequality. For $q \in \etanet$, let $\calE_q$ denote the event that no element in $S$ lands within distance $1-\eta$ of $q$ and $\calE$ denote the union of events $\calE_q$ over $q \in \etanet$.  Now, 
    \[
    \Pr[\calE_q] 
    = \left( 1 - \frac{\vol((1-\eta) K)}{\vol(\torus)} \right)^s 
    \leq \left( 1 - (1-\eta)^d \frac{\vol(K)}{\vol(\torus)} \right)^{\lambda \frac{\vol(\torus)}{\vol(K)}}
    \leq e^{-\lambda (1-\eta)^d} 
    \leq e^{-\lambda/100}.
    \]
    We can then union bound over $\etanet$ to conclude that $\calE$ occurs with small probability.  Indeed,
    \[
    \Pr[\calE]
    \leq \sum_{q \in \etanet} \Pr[\calE_q]
    \leq |\etanet| \cdot e^{-\lambda/100}
    \leq 2^{O(d \log d)} \cdot e^{-\lambda/100}
    \leq 0.01,
    \]
    for a sufficiently large constant $C$.  From now on, we 
    assume that the covering property holds.

    \paragraph{Secludedness.}  Suppose the secludedness condition is violated for some fixed $\ell$, i.e., there exists a point $x \in \torus$ such that the ball of radius $1/\ell$ around $x$ intersects more than $6C \cdot d \log(d+1) \cdot e^{d/\ell}$ cells. By the covering property, this implies that the ball of radius $1+1/\ell$ around $x$ intersects more than $6C \cdot d \log(d+1) \cdot e^{d/\ell}$ points in $S$. Furthermore, since $x$ is within distance $\eta$ of some point $q$ in $\etanet$, we have that the ball of radius $r_\ell := 1+1/\ell+\eta$ around $q$ intersects more than $6C \cdot d \log(d+1) \cdot e^{d/\ell}$ points in $S$. It is therefore sufficient to upper bound the probability that there exists $q \in \etanet$ such that
    \[ F_q := \left| S \cap \left( q + r_\ell K \right)\right| > 6C \cdot d \log(d+1) \cdot e^{d/\ell} .\]
    Let $\calE_{\ell, q}$ be the event this happens.  Note that $F_q$ is the sum of $s$ i.i.d. Bernoulli random variables:
    \[ \mathbb{E}[F_q] 
    = s \frac{\vol(r_\ell K)}{\vol(\torus)}
    = s \frac{r_\ell^d \cdot \vol(K)}{\vol(\torus)}
    = r_\ell^d \cdot \lambda . \]
    Next, we bound $r_\ell^d$.  Clearly
    $r_\ell^d \geq 1$, while 
    \[
    r_\ell^d 
    = \left( 1 + \frac{1}{\ell} + \eta \right)^d
    = \left( 1 + \frac{1}{\ell} + \frac{1}{4d} \right)^d
    \leq e^{\left( \frac{1}{\ell} + \frac{1}{4d} \right) \cdot d}
    \leq 2e^{d/\ell}.
    \]
    Using a simple Chernoff bound,
    \[ 
    \Pr[\calE_{\ell, q}] = \Pr[F_q > 6C \cdot d \log(d+1) \cdot e^{d/\ell})] \leq  \Pr[F_q > 3 \lambda \cdot r_\ell^d] \leq 2^{-\lambda \cdot r_\ell^d} \leq 2^{-\lambda},\]
    where we used $6C \cdot d \log(d+1) \cdot e^{d/\ell} \geq 3 \lambda r_\ell^d$.
    We can then union bound over $\etanet$ to conclude that $\calE = \cup_{\ell, q} \calE_{\ell, q}$ occurs with small probability.  Indeed,
     \[ 
     \Pr[\calE] \leq \sum_{\ell \in [d]} \sum_{q \in \etanet} \Pr[\calE_{\ell, q}] 
     \leq d \cdot |\etanet| \cdot 2^{-\lambda}
    = d \cdot 2^{O(d \log d)} \cdot 2^{-C \cdot d \log(d+1)}
    \leq 0.01. \]

By the probabilistic method, a covering and secluded set $S$ of size $s$ exists for $\torus$. We can now create such a set $\msspset$ for all of $\reals^d$ simply by letting 
$\msspset = S + L \integers^d$.
Since $S + K = \torus$, $\msspset + K = \reals^d$.  Thus, every point in $\reals^d$ is within $\| \cdot \|_{K}$ distance 1 of some point in $\msspset$.  The Voronoi cells of each point in $\msspset$ with respect to
$\|\cdot\|_K$-norm yields the MSSP partition
$\calP$ (with scaling factor $\rho = 1$).
\end{proof}
\begin{theorem}
    \label{thm:k_norm_part}
    Let $K \subset \reals^d$ be any symmetric set such that $\| \cdot \|_K$ is a norm. There exists a $k_d$-MSSP for $\| \cdot \|_{K}$, 
    where $k_d(\ell) = O(d \log d \cdot e^{ d/ \ell})$.
\end{theorem}
\begin{proof}
    By John's theorem (\cref{thm:john}), there exists an invertible linear transformation $T: \reals^d \rightarrow \reals^d$ such that if we let $\tilde{K} = T(K)$, then
    $\calB^d_2(1) \subseteq \tilde{K} \subseteq \calB_2^d(\sqrt{d}).$
    Observe that $z \in t K$ if and only if $T(z) \in T(tK) = t \cdot T(K)$. Therefore, $\| x \|_K = \| T(x) \|_{\tilde{K}}$, so $T$ preserves distances. 

    By \cref{lem:john_K}, there is 
    a $k_d$-MSSP
$\widetilde{\calP}$ for $\|\cdot\|_{\widetilde K}$, where
$k_d(\ell)=O(d \log d \cdot e^{d/\ell}).$  

Define a partition of $\reals^d$ by
$\calP = \{ T^{-1}(\widetilde P) \mid 
          \widetilde P \in \widetilde\calP \}$
with the representative function $r(T^{-1}(\widetilde P)) = T^{-1} \widetilde r(\widetilde P)$.
We now show that $\calP$ has the two desired properties. First, if $P=T^{-1}(\widetilde P)\in\calP$ and $x\in P$, then $Tx \in \widetilde P$, hence $\|x - r(P)\|_K = \|T(x) - T(r(P))\|_{\widetilde K} \le 1$.

Next, we show the secludedness property. Consider any point $v \in \reals^d$ and let us consider all the partition cells of $\mathcal{P}$ intersected by $v + (1/\ell)K$, in other words, those within distance $1/\ell$ of $v$ in the $\| \cdot \|_K$ norm. It suffices to upper bound
    \begin{align*}
        \Big| \{ P \in \mathcal{P} : P \cap \left( v + \frac{1}{\ell} K \right) \neq \emptyset \}\Big|
        &= \Big| \{ P \in \mathcal{P} : T(P) \cap T\left( v + \frac{1}{\ell} K \right) \neq \emptyset \}\Big|\\
        &= \Big| \{ \tilde{P} \in \tilde{\mathcal{P}} : 
        \tilde{P} \cap \left( T(v) + \frac{1}{\ell} \tilde{K} \right) \neq \emptyset \}\Big| 
        \leq O(d \log d \cdot e^{d/\ell}).
    \qedhere
    \end{align*}
\end{proof}
We remark that our construction of an MSSP for $\| \cdot \|_K$ is essentially the same as that of random coverings studied by Rogers~\cite{Rogers_1957}, who considered the question of how to cover $\R^d$ using $K$ with minimum \emph{density}. Our construction can be viewed as placing a copy of $K$, centered at all points in $\msspset$. The fact that $\msspset + K = \R^d$ ensures that this is a covering of the space. However, our objective is different. A more related objective comes from the follow up work of Erd\H{o}s and Rogers~\cite{Erdos1962}, who aimed to bound the \emph{multiplicity} (i.e., number of times it is covered) of any point in $\reals^d$. This is similar to our objective when the threshold distance $\gamma \cdot \rho$ approaches zero. However, it still does not account for the case of general (and multiple) values of $\gamma$. Thus, we are not aware of any method that recovers an MSSP directly from prior work.

\subsection{\boldmath Explicit Construction for the \texorpdfstring{$\ell_\infty$}{L-infty}-norm}
In this section, we give an explicit construction of a rounding scheme that implicitly defines an MSSP induced by the pre-image of the rounding scheme. Though this scheme has slightly worse parameters than those of \cref{subsec:k-norm}, it has the advantage of being explicit and has an efficient rounding function. Indeed, the rounding function can be implemented in polynomial time, and this will be key to our construction of a polynomial time, low randomness mechanism later in \cref{sec:efficient}.

\subsubsection{Rounding Scheme and Properties} 
We first define a rounding scheme $\Round$ (\cref{alg: inf_rounding}), which rounds points in $\reals^d$ to a representative point. We can then define the partition induced by $\Round$ by partitioning all $x \in \reals^d$ according to the points to which they would round under $\Round$. See \cref{fig:secluded_partition} for a visualization of the partition induced by \cref{alg: inf_rounding} for $d=2$.  We use $\fr{y} = y - \floor{y}$ to denote the fractional part of $y$.

\begin{figure}[htbp]
\begin{center}
\begin{tikzpicture}[scale=1.8]

\foreach \i in {-2,-1,0,1} {
    \foreach \j in {-2,-1,0,1} {
        
\pgfmathsetmacro{\rL}{0.6 + 0.3*sin(\i*70 + \j*40)}
        \pgfmathsetmacro{\gL}{0.6 + 0.3*sin(\i*40 + \j*80 + 120)}
        \pgfmathsetmacro{\bL}{0.6 + 0.3*sin(\i*80 + \j*30 + 240)}
        \definecolor{colorL}{rgb}{\rL, \gL, \bL}

\filldraw[fill=colorL, draw=white, thick]
            (\i, \j) --
            (\i+1, \j) --
            (\i+1, \j+0.5) --
            (\i+0.5, \j+0.5) --
            (\i+0.5, \j+1) --
            (\i, \j+1) -- cycle;

\pgfmathsetmacro{\rS}{0.6 + 0.3*sin(\i*70 + \j*40 + 180)}
        \pgfmathsetmacro{\gS}{0.6 + 0.3*sin(\i*40 + \j*80 + 300)}
        \pgfmathsetmacro{\bS}{0.6 + 0.3*sin(\i*80 + \j*30 + 60)}
        \definecolor{colorS}{rgb}{\rS, \gS, \bS}

\filldraw[fill=colorS, draw=white, thick]
            (\i+0.5, \j+0.5) rectangle (\i+1, \j+1);

\filldraw[black!80] (\i, \j) circle (0.015);
        \filldraw[black!80] (\i+0.5, \j+0.5) circle (0.015);
    }
}

\draw[thick, ->, >=stealth] (-2.2, 0) -- (2.2, 0) node[right] {$x$};
\draw[thick, ->, >=stealth] (0, -2.2) -- (0, 2.2) node[above] {$y$};

\foreach \x in {-2, -1, 1, 2}
    \draw (\x, 1pt) -- (\x, -1pt) node[anchor=north, font=\small] {$\x$};
\foreach \y in {-2, -1, 1, 2}
    \draw (1pt, \y) -- (-1pt, \y) node[anchor=east, font=\small] {$\y$};
    
\node[anchor=north east, font=\small] at (-1pt, -1pt) {$0$};

\filldraw[black] (-0.8, -0.6) circle (0.015) node[right, font=\footnotesize] {$x_1$};
\draw[->, >=stealth, thick, dashed, shorten >=2pt] (-0.8, -0.6) -- (-1, -1) 
    node[midway, sloped, above, fill=white, fill opacity=0, text opacity=1, inner sep=1pt, font=\scriptsize, rounded corners=1pt] {$R(x_1)$};

\filldraw[black] (1.75, 0.8) circle (0.015) node[right, font=\footnotesize] {$x_2$};
\draw[->, >=stealth, thick, dashed, shorten >=2pt] (1.75, 0.8) -- (1.5, 0.5) 
    node[midway, sloped, above, fill=white, fill opacity=0, text opacity=1, inner sep=1pt, font=\scriptsize, rounded corners=1pt] {$R(x_2)$};

\filldraw[black] (-1.2, 1.25) circle (0.015) node[above right=-2pt, font=\footnotesize] {$x_3$};
\draw[->, >=stealth, thick, dashed, shorten >=2pt] (-1.2, 1.25) -- (-2, 1) 
    node[midway, sloped, above, fill=white, fill opacity=0, text opacity=1, inner sep=1pt, font=\scriptsize, rounded corners=1pt] {$R(x_3)$};

\end{tikzpicture}
\end{center}
    \caption{The induced two-dimensional $\ell_\infty$-MSSP from \cref{alg: inf_rounding}. Included are example points $x_1, x_2, x_3$ along with where they would be rounded under $\Round$.}
    \label{fig:secluded_partition}
\end{figure}
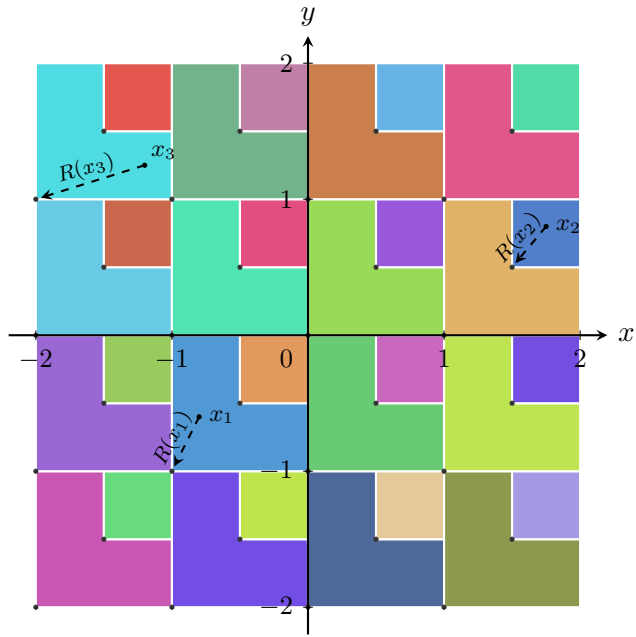

\begin{algorithm}[H]
\caption{$\ell_\infty$ Rounding Scheme $\Round$}
\label{alg: inf_rounding}
\begin{algorithmic}
    \State \textbf{Parameters:} $d \in \integers_{> 0}$
    \State \textbf{Input:} Point $x \in \reals^d$
    \Statex
    \State For each $a \in \{ 0, \dots, d-1 \}$, compute $\Phi_a(x) = \sum_{i=1}^d \{ x_i - \frac{a}{d} \}$

    \State $A(x) \leftarrow \argmin_{a} \Phi_a(x)$ \Comment{break ties by choosing smallest such $a$}

    \For {$i \in [d]$}
        \State $y_i \leftarrow \frac{A(x)}{d} + \left\lfloor x_i - \frac{A(x)}{d} \right\rfloor$
    \EndFor\\
    \Return $(y_1, \dots, y_d)$
\end{algorithmic}
\end{algorithm}

We first show that the $\ell_\infty$-error due to rounding is small.
\begin{lemma}
    \label{lem: inf_movement_bound}
    For all $d \in \integers_{> 0}$ and $x \in \reals^d$, $\| \Round(x) - x \|_\infty \leq 1$.
\end{lemma}
\begin{proof}
    We will show that for each $i \in [d]$, $|y_i - x_i| \leq 1$. Note that for all $0 \leq z < 1$, $x_i - 1 \leq \lfloor x_i - z \rfloor + z \leq x_i$. Setting $z = A(x)/d$ confirms $|y_i - x_i| \leq 1$.
\end{proof}
We now prove the key properties underlying the fact that $\Round$ induces a good MSSP: that if $A(x) = a$, then not many coordinates $i$ can have $\{ x_i - a/d\}$ close to $1$.  For the remainder of this section, we fix $d \in \integers_{> 0}$, $x \in \reals^d$ and use the notation from \cref{alg: inf_rounding}.
\begin{lemma}
    \label{lem:phi_diff}
    For all $a, b \in \{0, \ldots, d-1\}$, if $b \equiv a - m \pmod d$ for some $m \in \{0, \dots, d-1 \}$, then
    \[ \Phi_b(x) - \Phi_a(x) = m - \left| \left\{ i : \fr{x_i - \frac{a}{d}} \ge 1 - \frac{m}{d} \right\} \right| .\]
    
\end{lemma}
\begin{proof}
    Let $t_i = \fr{x_i - \frac{a}{d}} \in [0,1)$.  By periodicity, 
    \[
    \fr{x_i - \frac{b}{d}} 
    = \fr{t_i + \frac{m}{d}}
    = t_i + \frac{m}{d} - \one_{t_i \ge 1 - m/d},
    \]
    where $\one$ denotes the  indicator function.  Now, the change in $\Phi$ is exactly
    \begin{align*}
        \Phi_b(x) - \Phi_a(x) 
      &= \sum_{i=1}^d \left( \left\{ x_i - \frac{b}{d} \right\} - \left\{ x_i - \frac{a}{d} \right\} \right)
      = \sum_{i=1}^d \left( \fr{t_i + \frac{m}{d}} - t_i \right)\\ 
      &= \sum_{i=1}^d \left( \frac{m}{d} - \one_{t_i \ge 1 - \frac{m}{d}} \right) 
      = m - \left| \left\{ i : t_i \ge 1 - \frac{m}{d} \right\} \right|.
      \qedhere
    \end{align*}
\end{proof}
\begin{lemma}
\label{lem:balancing}
For every \(\theta\in[0,1]\),
\[
  \left|
  \left\{
  i\in [d]:
  \fr{x_i-\frac{A(x)}{d}}\ge 1-\theta
  \right\}
  \right|
  \le
  d\theta+1.
\]
\end{lemma}
\begin{proof}
Let $a = A(x)$ and let $t_i = \fr{x_i - \frac{a}{d}} \in [0,1)$. Consider any $b, m \in [0, d-1]$ such that $b \equiv a - m \pmod d$. Since $a$ minimizes $\Phi_a(x)$, we have $\Phi_b(x) - \Phi_a(x) \geq 0$. By \cref{lem:phi_diff},
\[ m - \left| \left\{ i : t_i \ge 1 - \frac{m}{d} \right\} \right| = \Phi_b(x) - \Phi_a(x) \geq 0 .\]
This yields $\left| \left\{ i : t_i \ge 1 - \frac{m}{d} \right\} \right| \le m$. For any $\theta \in [0,1]$, choose $m = \ceil{d\theta}$. Note that if $m = d$, the lemma is trivially true since this would imply $d \theta+1 \geq d$. Therefore, we assume $m \in [0, d-1]$. Since $\frac{m}{d} \ge \theta$, any coordinate $i$ with $t_i \ge 1 - \theta$ also satisfies $t_i \ge 1 - \frac{m}{d}$. Thus, the number of such coordinates is at most $m = \ceil{d\theta} \le d\theta + 1$.
\end{proof}

\subsubsection{Bounding Intersected Partition Cells}

We now prove that any sufficiently small ball only intersects relatively few partition cells of the partition induced by $\Round$, satisfying 
the second condition of an MSSP.  In this section, let $\calB_\infty(x; r)$ denote the $\ell_\infty$-ball of radius $r$ centered at $x$.

\begin{lemma}
    \label{lem: inf_intersect_bound}
    Fix $d \in \integers_{> 0}$, $x \in \reals^d$, and define $\Round$ as in \cref{alg: inf_rounding}. For every $\ell \in [d]$,
    \[ |\Round(\calB_\infty(x; 1/\ell)) | \leq (\ell+1)^{O(d/\ell)}. \]
\end{lemma}
\begin{proof}
Let $r=1/\ell$. For any fixed shift $a \in [0, d-1]$ for any point $z \in \calB_\infty(x; r)$, the shifted values $z_i - a/d$ lie in an interval of length $2r = 2/\ell$. We handle two cases based on this length.

(i) \emph{Small $\ell$: $\ell \le 2$.} 
Since the interval length is $2/\ell \le 2$, for any fixed $a$, there are at most three possible values for $\frac{a}{d} + \lfloor z_i - \frac{a}{d} \rfloor$ over $z \in \calB_\infty(x; r)$. So across the $d$ possible shifts $a$, the total number of vectors in $\Round(\calB_\infty(x; 1/\ell))$ is at most $d \cdot 3^d$. 
Because $\ell \le 2$, we have $d \le 2(d/\ell)$, so $d \cdot 3^d \le 2^{O(d/\ell)} \le (\ell+1)^{O(d/\ell)}$.

(ii) \emph{Large $\ell$: $\ell \ge 3$.}
Here, the interval length $2r = 2/\ell < 1$. Therefore, for each fixed shift $a$, over $z \in \calB_\infty(x; r)$, after rounding, each coordinate takes at most two adjacent values for $\frac{a}{d} + \lfloor z_i - \frac{a}{d} \rfloor$. Let $M_i(a) = \max_{y \in \calB_\infty(x; r)} \floor{y_i - a/d}$.

Consider a point $z \in \calB_\infty(x; r)$ that selects shift $A(z) = a$. Say $\floor{z_i - a/d}$ does not take the maximum floor value $M_i(a)$. Then, we must have $\floor{z_i - a/d} = M_i(a) - 1$. Furthermore, by the definition of $M_i(a)$, $\floor{x_i + 1/\ell -a/d} = M_i(a)$, which implies $x_i + 1/\ell - a/d \geq M_i(a)$. Since $z_i + 1/\ell \geq x_i$, we have $z_i + 2/\ell -a/d \geq M_i(a)$. Rearranging leaves us with $z_i - a/d \geq M_i(a) - 2/\ell$. This, together with the fact that $\floor{z_i - a/d} = M_i(a)-1$, implies that the fractional part of $z_i - a/d$ is tightly bounded:
\[
  \fr{z_i - \frac{a}{d}} = \left(z_i - \frac{a}{d}\right) - (M_i(a) - 1) \ge 1 - \frac{2}{\ell}.
\]
So we have established that if $\floor{z_i - a/d} = M_i(a)-1$, then $\{ z_i - a/d \} \geq 1-2/\ell$. So by \cref{lem:balancing}
with $\theta = 2/\ell$, for any $z \in \calB_{\infty}(x; r)$, there exist at most $k = \floor{2d/\ell} + 1$ coordinates $i \in [d]$ such that $\lfloor z_i - a/d \rfloor = M_i(a) - 1$. 

Thus, $\Round(z)$ is completely determined by $a$ and this subset of at most $k$ coordinates. Since $3 \leq \ell  \le d$, we have $k \le 2d/\ell + d/\ell = 3d/\ell \le d$. Summing over all $d$ shifts and using the standard binomial tail bound $\sum_{j=0}^k \binom{d}{j} \le \left(\frac{ed}{k}\right)^k$, the number of values $\Round(z)$ takes on over $z \in \calB_{\infty}(x; r)$ is at most
\[
  |\Round(\calB_\infty(x; r))| 
  \le d \sum_{j=0}^k \binom{d}{j} 
  \leq d \sum_{j=0}^{3d/\ell} \binom{d}{j} 
  \le d \left(\frac{ed}{3d/\ell}\right)^{3d/\ell} 
  = d \left(\frac{e\ell}{3}\right)^{3d/\ell} 
  \le d \ell^{100 d/\ell} = (\ell+1)^{O(d/\ell)}.
\qedhere
\]
\end{proof}
By combining
 \cref{lem: inf_movement_bound} and \cref{lem: inf_intersect_bound}, we conclude:
\begin{theorem}\label{thm: inf_main}
For any \(d \in \integers_{> 0} \), the partition induced by \cref{alg: inf_rounding} is a $k_d$-MSSP for $k_d(\ell) = (\ell+1)^{O(d/\ell)}$.
\end{theorem}

\subsection{Tail Bounds from MSSPs}

Our definition of an MSSP with scaling factor $\rho$ only tells us how many partition cells can intersect a ball of discrete radius $\rho/\ell$ for $\ell \in [d]$.  It does not bound the number of partition cells a ball of real-valued radius more than $\rho/d$ may intersect. Although the definition of an MSSP does not bound this quantity, we can bound it by using the properties of an MSSP. 

Recall that $\Delta_P(x) := \inf_{z \in P} \| x - z\|$. For a fixed $y \in \reals^d$, let 
\[
\cP_{< \gamma}(y) = \{ P \in \calP \mid \Delta_P(y) < \gamma \cdot \rho \}.
\]
The case $\gamma \in [1/d, 1]$ is easy.
\begin{lemma}
\label{lem:extending_partitions_small}
     Let $\calP$ be a $k_d$-MSSP for $k_d(\ell) = m(d) \cdot e^{\lambda(d)/\ell}$, where
    $m(\cdot)$ and $\lambda(\cdot)$ are functions.
     For any $y \in \reals^d$ and for any $\gamma \in [1/d, 1]$, $|\cP_{< \gamma}(y)| \leq m(d) \cdot e^{2 \lambda(d) \cdot \gamma}$.
\end{lemma}
\begin{proof}
  Since $\gamma \leq 1$, we have $1/(2\gamma) \leq \lfloor 1/ \gamma \rfloor := \ell$. Furthermore, $1/\ell \ge \gamma$. For any $y$,
    \[ |\cP_{< \gamma}(y)| 
    ~\le~ |\cP_{<1/\ell}(y)|
    ~\le~ m(d) \cdot e^{\frac{\lambda(d)}{\ell}}
    ~\le~ m(d) \cdot e^{2 \lambda(d) \cdot \gamma}. \qedhere\]
\end{proof}
Next we address the case $\gamma > 1$.
\begin{lemma}
    \label{lem:extending_partitions_large}
    If $\calP$ is a $k_d$-MSSP with respect to $\| \cdot \|_K$, then 
    for any $y \in \reals^d$ and $\gamma \geq 1$, 
    $|\cP_{< \gamma}(y)| \leq k_d(1) \cdot (3 \gamma)^d$.
\end{lemma}
\begin{proof}
Assume without loss of generality that $y = \vec{0}$. We may also rescale the grid to $\rho = 1$ without affecting $|\cP_{< \gamma}(y)|$, so assume $\rho = 1$. Let $\delta = 1$, $B_1 = \gamma K$, and $B_2 = K$. By \cref{lem: covering_num_bound}, $\gamma K$ can be covered by a $\delta$-net of size at most $3^d \gamma^d$. 
    Now place a $\| \cdot \|_K$ ball of radius $\rho$ around each $\delta$-net point. 
    Furthermore, the MSSP bound guarantees that each ball of radius $\rho$ intersects at most $k_d(1)$ cells. Summing over all of the points in the $\delta$-net, $\gamma K$ intersects at most $k_d(1) \cdot 3^d \gamma^d$ partition cells.
\end{proof}

We now show the useful property of good MSSPs with proper scaling $\rho$: if a partition cell $P$ is sampled
 with probability proportional to $\exp(-\frac{\varepsilon}{2} s_P(v))$ for 
 any upper bound $s_P(v) \geq \Delta_P(v)$, then the probability that $P$ is at a distance greater than $\gamma$ from $v$ decays exponentially with $\gamma$. 
\begin{lemma} \label{lem:tail_bound}
    Let $\varepsilon > 0, \alpha \geq 1, M \geq 1, A >0$. Let $A' = \max(2A, A+ 3)$ and $C = 2A' + 2$. Let $\calP$ be a $(k_d, \rho)$-MSSP for $\rho = \frac{C \alpha d}{\varepsilon}$ where $k_d(\ell) \leq M e^{A \alpha d/\ell}$. Let $v \in \mathbb{R}^d$ be arbitrary and $s_P(\cdot)$ be a nonnegative function such that $s_P(v) \geq \Delta_P(v)$ for all $P \in \mathcal{P}$ and $s_{P_0}(v) = 0$ for some $P_0$. Consider sampling $P \in \mathcal{P}$ with probability proportional to $\exp(-\frac{\varepsilon}{2} s_P(v))$. Then, 
    \[ \Pr[P \notin \cP_{< \gamma}] \leq 2 M \cdot e^{(2A+1) \alpha - \alpha d \gamma}, \mbox{ for } \gamma \geq 1/d. \]
\end{lemma}

\begin{proof}   
        Let $k = \lfloor \gamma d \rfloor \ge 1$. We bound the tail probability by summing over rings of step size $1/d$:
    \begin{align*}
        \Pr[P \notin \cP_{< \gamma}] 
        \leq \sum_{j=k}^{\infty} \Pr\left[P \in \cP_{< \frac{j+1}{d}} \setminus \cP_{< \frac{j}{d}}\right]
        \leq \sum_{j=k}^\infty |\cP_{< \frac{j+1}{d}}| \cdot \exp \left(-\frac{\varepsilon}{2} \cdot \frac{j}{d} \rho \right)
        = \sum_{j=k}^\infty |\cP_{< \frac{j+1}{d}}| \cdot \exp \left(-\frac{C \alpha}{2} j \right),
    \end{align*}
    where the second inequality follows by observing that, for any ring $j$, the probability of sampling $P$ is at most $\exp(-\frac{\varepsilon}{2} s_P(v))$ (since $s_{P_0}(v) = 0$ causes the normalizing factor to be at least 1), which is at most $\exp(-\frac{\varepsilon}{2} \Delta_P(v))$ since $s_P(v) \geq \Delta_P(v)$, and $\Delta_P(v) \geq \frac{j}{d} \rho$ by assumption. The final equality follows from our choice of $\rho$.
    
    Let $A$ and $B$ be universal constants such that $k_d(\ell) \leq M e^{A \alpha d/ \ell}$. We bound $|\cP_{< \frac{j+1}{d}}|$ in two cases:

        (i) \emph{$j < d$.}  By \cref{lem:extending_partitions_small}, $|\cP_{< \frac{j+1}{d}}| \le M e^{2 A \alpha d \frac{j+1}{d}} = M e^{2 A \alpha (j+1)} = M e^{2 A \alpha} e^{2 A \alpha j}$.

        (ii) \emph{$j \geq d$.} By \cref{lem:extending_partitions_large}, $|\cP_{< \frac{j+1}{d}}| \le k_d(1) \left( 3 \frac{j+1}{d} \right)^d \leq M e^{A \alpha d} \left( 3 \frac{j+1}{d} \right)^d \leq M e^{A \alpha d} \left( 10 \frac{j}{d} \right)^d $. 
        Since $\ln(10 x) \le 3x$ for all $x \ge 1$, we have $d \ln(10 j/d) \le 3j \le 3\alpha j$. Thus, $(10 \frac{j}{d})^d \le e^{3 \alpha j}$, so $|\cP_{< \frac{j+1}{d}}| \leq M e^{A \alpha d} e^{3 \alpha j}$. Combining this with $A \alpha d \le A \alpha j$, $|\cP_{< \frac{j+1}{d}}|$ is strictly bounded by $M e^{A' \alpha j}$.
        
    Since $e^{2 A \alpha} > 1$, in both cases, $|\cP_{< \frac{j+1}{d}}|$ is upper bounded by $M e^{2 A \alpha} e^{A' \alpha j}$. 
    Therefore, by substituting $C = 2A' + 2$,
    \[ |\cP_{< \frac{j+1}{d}}| \cdot \exp \left(-\frac{C \alpha}{2} j \right)
    \leq 
    M e^{2 A \alpha} e^{A' \alpha j} \cdot e^{-(A'+1) \alpha j}
    = M e^{2 A \alpha} e^{- \alpha j},\]
    where the final inequality follows by choosing $C$ sufficiently large.
    Plugging this into our original bound on $\Pr[P \notin \cP_{< \gamma}]$ yields a converging geometric series:
    \[ \Pr[P \notin \cP_{< \gamma}] \leq \sum_{j=k}^\infty M e^{2 A \alpha} e^{- \alpha j} 
    \leq M e^{2 A \alpha} \frac{e^{-\alpha k}}{1 - e^{-\alpha}}
    \leq M e^{2 A \alpha} \frac{e^{-\alpha k}}{1 - e^{-1}}
    \leq 2 M e^{2 A \alpha} e^{-\alpha k}. 
\]
    Since $k = \lfloor \gamma d \rfloor \geq \gamma d - 1$, we have $e^{-\alpha k} \leq e^{-\alpha (\gamma d -1)} = e^{\alpha} e^{-\alpha \gamma d}$. Subtituting this into the above, we obtain
    \[ \Pr[P \notin \cP_{< \gamma}]
    \leq 2 M e^{2 A \alpha} e^{-\alpha k}
    \leq 2M e^{2A \alpha} e^{\alpha} e^{-\alpha \gamma d}
    = 2 M e^{(2A+1) \alpha-\alpha \gamma d}. 
    \qedhere
    \]

\end{proof} \section{Low Randomness Mechanism for any $\|\cdot\|_K$-norm}

In this section, we use 
an MSSP partition to construct a low randomness $\varepsilon$-DP mechanism for any $\| \cdot \|_K$-norm.  To this end, we use
the $k_d$-MSSP partition $\calP$ with respect to the $\|\cdot\|_K$-norm from~\cref{thm:k_norm_part}, where $k_d(\ell) = O(d \log d \cdot e^{d/\ell})$; we set the scaling factor $\rho = C d/\varepsilon$ for a
large constant $C$.  Let 
$r: \calP \rightarrow \mathbb{R}^d$ be the representative
function.  

\begin{theorem}
    \label{thm:k_norm_to_mech}
Let $\varepsilon > 0$.  
Let $f: \mathcal{X}^* \rightarrow \reals^d$ be such that
$\sens_{\| \cdot \|_K}(f) \leq 1$.  There is a mechanism that (i) is $\varepsilon$-DP,
(ii) has expected $\| \cdot \|_K$-error  $O(d/\varepsilon)$, and
(iii) has randomness complexity $O(\log d)$.
\end{theorem}
This immediately implies low randomness mechanisms for linear queries in
the $\ell_p$-norm.
\begin{corollary} \label{cor:lp-mech}
    Let $p \geq 1$ and let $\varepsilon > 0$.  Let $f: \mathcal{X}^* \rightarrow \mathbb{R}^d$ be such that $\sens_{\ell_p}(f) \leq 1$.  There is a mechanism that (i) 
    is $\varepsilon$-DP, (ii) 
    has expected $\ell_p$-error  $O(d/\varepsilon)$, and 
    (iii) has randomness complexity $O(\log d)$.
\end{corollary}
The mechanism in these results is 
described in \Cref{alg:proportional_mech}.  In the 
rest of the section, we will prove the three properties of this mechanism, namely, privacy, utility, and randomness complexity.

\begin{algorithm}[h]
\caption{$\ExpPart$: Exponential Partition Sampling Mechanism}
\label{alg:proportional_mech}
\begin{algorithmic}
\State \textbf{Input: } Dataset $X \in \cX^*$
\State \textbf{Parameters:} Norm $\| \cdot \|_K$, function $f: \cX^* \rightarrow \mathbb{R}^d$, MSSP $\cP$, Privacy parameter $\varepsilon$, Representative function $\repr: \cP \to \R^d$ satisfying $r(P) \in P$ for all $P \in \mathcal{P}$\\

\State Sample $P \in \cP$ proportional to $\exp(-\frac{\varepsilon}{2} \cdot \Delta_P(f(X)))$ 
\Comment{$\Delta_P(x) = \inf_{z \in P} \| x - z\|_K$}

\noindent \Return $\repr(P)$

\end{algorithmic}
\end{algorithm}

\subsection{Privacy}

We start by analyzing the privacy guarantee of \Cref{alg:proportional_mech}, which follows almost immediately from the definition and the privacy guarantee of the exponential mechanism.

\begin{lemma} \label{lem:main-dp}
$\Delta_P(f(\cdot))$ has sensitivity at most $1$ for any fixed $P$; consequently, \Cref{alg:proportional_mech} is $\eps$-DP.
\end{lemma}

\begin{proof}
Notice that \Cref{alg:proportional_mech} is an instantiation of the exponential mechanism with $\scr_P(X) := \Delta_P(f(X))$. By \Cref{thm:exm-dp}, it suffices to show that $\scr_P$, i.e., $\Delta_P \circ f$, has sensitivity at most $1$.  

First, for any two points $z_1, z_2 \in \R^d$, and $y \in P$, we have by the triangle inequality
\[
\|z_1 - y\|_K \leq \|z_1 - z_2\|_K + \|z_2 - y\|_K,
\] 
and by taking infimum over $y \in P$ on both sides and using the definition
of $\Delta_P$, we have
\[
\Delta_P(z_1) = \inf_{y \in P} \|z_1 - y \|_K 
\leq \|z_1 - z_2 \|_K + \inf_{y \in P} \|z_2 - y \|_K 
= \|z_1 - z_2 \|_K + \Delta_P(z_2).  
\] 
Thus, $\Delta_P(z_1) - \Delta_P(z_2) \leq \|z_1 - z_2 \|_K$.  By symmetry,
$\Delta_P(z_2) - \Delta_P(z_1) \leq \|z_1 - z_2 \|_K$ as well, from which we
obtain $|\Delta_P(z_1) - \Delta_P(z_2)| \leq \|z_1 - z_2\|_K$, i.e., 
$\Delta_P$ is Lipschitz.  

Now, consider any pair $X \sim X'$ of neighboring datasets.  We have
$|\Delta_P(f(X)) - \Delta_P(f(X'))| \leq \|f(X) - f(X')\|_K
\leq 1$ from the Lipschitzness of $\Delta_P$ and since
$\sens_{\|\cdot\|_K}(f) \leq 1$.  
\end{proof}

\subsection{Utility}

To analyze the expected utility of \cref{alg:proportional_mech}, we note that $\ExpPart$ has two sources of error. The first comes from moving to a partition $P$ rather than outputting $f(x)$. This induces an error that is bounded by $\Delta_P(f(x))$, where $P$ is a random variable. We will show that $\Delta_P(f(x))$ is fairly small in expectation since \cref{lem:tail_bound} guarantees that our mechanism usually outputs representative points of fairly close partition cells. The second source of error comes from $r(\cdot)$, which chooses a representative point within $P$; this error is bounded by the
scaling parameter $\rho$ that we chose.

\begin{lemma}
    \label{lem:utility}
      $\expect[\Delta_P(f(x))] = O\!\left(\frac{\log d}{\varepsilon} \right)$; consequently, \cref{alg:proportional_mech} has expected $\| \cdot \|_K$-error $O\left(\frac{ d}{\varepsilon}\right)$.
\end{lemma}

\begin{proof}
By \cref{lem:tail_bound}, for 
some universal constants $c_1, c_2$ and
$c_3 = c_2/C$, we have
\[ 
\Pr[ P \notin \cP_{< (k-1)/\rho}]
    \leq 
d^{c_1} \cdot \exp\left(-c_2 d \frac{k-1}{C d / \varepsilon}\right) 
    = d^{c_1} \cdot \exp\left(-c_3  \varepsilon (k-1)\right). \]
Let $k^* = \max\left(\lceil \frac{\rho}{d} \rceil + 1, \lceil \frac{c_1 \ln d}{c_3 \varepsilon} \rceil + 1\right)$. 
    Since $\frac{\rho}{d} = \frac{C}{\varepsilon}$, we have $k^* = O\left(\frac{\log d}{\varepsilon}\right)$.  Now, 
    note that the choice of $k^*$ ensures that $e^{-c_3 \varepsilon (k^*-1)} \le 1/d^{c_1}$ and
    for $k > k^*$, we have $(k-1)/\rho \ge 1/d$.
\begin{align*}
\expect[\Delta_P(f(x))]
& \leq \sum_{k = 0}^\infty \Pr[\lceil \Delta_P(f(x)) \rceil > k] 
\leq \sum_{k = 0}^\infty \Pr[\Delta_P(f(x)) > k-1] 
\leq \sum_{k = 0}^\infty \Pr[ P \notin \cP_{< (k-1)/\rho}] \\
& \leq \sum_{k=0}^{k^*} 1 + \sum_{k=k^*+1}^\infty d^{c_1} \cdot e^{-c_3 \varepsilon (k-1)} 
        \le k^* + 1 + \sum_{j=0}^\infty e^{-c_3 \varepsilon j}
= k^* + 1 + \frac{1}{1 - e^{-c_3 \varepsilon}}\\
& \leq O\left( \frac{\log d}{\varepsilon} \right) + O\left( \frac{1}{\varepsilon} \right)
= O\left( \frac{\log d}{\varepsilon} \right).
\end{align*}
The total expected error is 
the sum of the maximum cell diameter ($2 \rho$) and the expected distance, i.e., $2 \rho + \expect[\Delta_P(f(x))] = O\left(\frac{d}{\varepsilon}\right)$, from our
choice of $\rho$. 
\end{proof}

\subsection{Randomness Complexity}

By \cref{thm:KY}, it suffices to bound the output entropy of our mechanism.  This in turn is upper bounded by the entropy of $P$, the random partition cell chosen by the mechanism.
\begin{lemma}
    \label{lem:randomness}
      $H(P) = O(\log d)$; consequently, \cref{alg:proportional_mech} has randomness complexity $O(\log d)$.
\end{lemma}
\begin{proof}
Let $Z = \lceil d \cdot \Delta_P(f(x))/\rho \rceil$. Note that if $Z = k$, then $P \in \cP_{< (k+1)/d}$.
    Applying \cref{lem:utility}, 
    \[ 
    \expect[Z] 
    = \expect \left[\left\lceil d \cdot \frac{\Delta_P(f(x)}{\rho} \right\rceil \right]
    = O \left(\frac{d}{\rho} \expect[\Delta_P(f(x))] \right)
    = O(\varepsilon) \cdot O\left(\frac{\log d}{\varepsilon} \right)
    = O(\log d).\]
We first bound $H(P \mid Z)$. Given $Z=k$, the mechanism outputs a cell from $\cP_{< (k+1)/d}$.  For $k \leq d$, 
the MSSP guarantee yields $|\cP_{< (k+1)/d}| \le O(d \log d \cdot e^{k+1})$, i.e., $\log |\cP_{< (k+1)/d}| = O(\log d) + O(k)$. For $k > d$, \cref{lem:extending_partitions_large}
yields $|\cP_{< (k+1)/d}| \leq d^{O(1)} \cdot e^{d} \cdot (3 (k+1)/d)^{d}$, i.e., $\log |\cP_{< (k+1)/d}| = O(d \log k)$. Thus,
    \begin{align*}
        H(P \mid Z) 
        &= \sum_{k=1}^\infty \Pr[Z=k] \cdot H(P \mid Z = k) \\
        &\leq \sum_{k=1}^\infty \Pr[Z=k] \cdot \log |\cP_{< (k+1)/d}| \\
        &\leq \sum_{k=1}^d \Pr[Z=k] \cdot \big(O(\log d) + O(k)\big) + \sum_{k=d+1}^\infty \Pr[Z=k] \cdot O(d \log k) \\
        &\leq O(\log d) \cdot \sum_{k=1}^\infty \Pr[Z=k] + O(1) \cdot \sum_{k=1}^\infty k \Pr[Z=k] + O(1) \\
        &= O(\log d) + O(\expect[Z]) + O(1) 
        = O(\log d).
    \end{align*}
    (The $O(1)$ follows because $\Pr[Z=k]$ decays exponentially as $e^{-\Omega(k)}$ by \cref{lem:tail_bound}, which dominates the polynomial $O(d \log k)$ term for $k > d$.)
    
    Next, we bound $H(Z)$. Using the well known information-theoretic fact that for any random variable $X$ with support on $\mathbb{Z}_{\geq 0}$, $H(X) \leq \expect[X] + O(1)$ \cite{golomb2013basic}, we 
    obtain $H(Z) \leq \expect[Z] + O(1) = O(\log d)$.  
    Therefore, $H(P) \leq H(Z) + H(P \mid Z) = O(\log d)$.
\end{proof}
\section{Efficient Mechanism for $\ell_\infty$-norm}
\label{sec:efficient}

Using the MSSP from \cref{thm: inf_main} and the same reasoning as in the proof of \cref{thm:k_norm_to_mech} implies a mechanism that uses $O(\log d)$ bits of randomness. However, the resulting mechanism is not immediately ``efficient''.  In this section, we show the following.
\begin{theorem}
\label{thm:efficient_main}
Let $\varepsilon \leq \log d$.
Let $f: \mathcal{X}^* \rightarrow \mathbb{Q}^d$ be such that
$\sens_{\| \cdot \|_{\infty}}(f) \leq 1$.  There is a mechanism that (i) is $\varepsilon$-DP,
(ii) has expected $\| \cdot \|_{\infty}$-error $O((d \log d)/\varepsilon)$,
(iii) has expected randomness complexity $O(\log d)$, and
(iv) has expected running time $\poly(T(f, X), 1/\varepsilon)$, where $T(f, X)$ denotes the time required to evaluate $f$ on input $X$.
\end{theorem}

Let $c = \lceil \log(4 e^{\varepsilon/2}/\varepsilon) \rceil$, $a = \lceil 2^c e^{-\varepsilon/2} \rceil $ and choose $\varepsilon'$ so that
\[ \lambda = e^{-\varepsilon'/2} = \frac{a}{2^c} .\]
Note that $\lambda$ can be represented in binary using only $c$ bits; also, it is easy to see that $\varepsilon/2 \leq \varepsilon' \leq \varepsilon$.  To make sampling computationally easier, we will show that our mechanism is $\varepsilon'$-DP.

Throughout this section, $\cP$ will refer to the partition induced by \cref{alg: inf_rounding} scaled so that $\rho = C \frac{d \log d}{\varepsilon}$ for some sufficiently large constant $C$.
Let $B = 7$, the constant in the exponent of $N_m$ from \cref{lem:efficient_enum}.
The high-level idea of our efficient mechanism is as follows. We wish to sample from a distribution with probabilities proportional to $\exp(-\frac{\varepsilon}{2} \Delta_P(f(X)))$ efficiently (actually, we sample from $\exp(- \frac{\varepsilon}{2} \lceil \Delta_P(f(X)) \rceil)$ for technical bit-complexity reasons). We will do so by first using some randomness to determine how large $\Delta_P(f(X))$ should be. This gives us an annulus from which we should sample $P$. We then enumerate the partitions in this annulus and use more randomness to select one among them.

\subsection{Efficiently Enumerating Partitions}

We first describe two algorithms that will be crucial to \cref{alg:efficient_sampler}. The first is for computing $\Delta_P(f(X)) = \inf_{y\in P}\|f(X)-y\|_\infty$ efficiently, where the challenge is due to the infimum.  We then use this to design a second algorithm that efficiently enumerates the partition cells that are close (for some definition of ``close'') to a particular point.

For $a \in \{ 0, \dots, d-1\}$ and $z \in \mathbb{Z}^d$, let $P_{a, z} \in \cP$  denote the partition cell whose representative point is $(z_1+a/d, \dots, z_d+a/d)$.
\begin{lemma}
\label{lem:dist_comp}
There is an algorithm
that given \(v\in\mathbb Q^d\) and $a \in \{0, \ldots, d-1\}$, $z \in \mathbb{Z}^d$, computes $\Delta_{P_{a, z}}(v)$ in time polynomial in the input length.

\end{lemma}
\begin{proof}
We use the notation from \cref{alg: inf_rounding}.
We now obtain a precise characterization of the points in
$P_{a, z}$.  It is easy to see that 
each $y \in P_{a, z}$ must be of the form 
\begin{align}
y_i = z_i + \frac{a}{d} + t_i \quad \text{where }0 \leq t_i < 1.
\label{eq:lem51}
\end{align}
Therefore, an exact characterization of $P_{a,z}$ is the set of points $y \in \mathbb{R}^d$ such that \eqref{eq:lem51} holds and $A(y) = a$. \cref{lem:phi_diff} tells us that the latter condition is exactly
    \begin{enumerate}
        \item For $b < a$, $\Phi_a(y) < \Phi_b(y)$ (equivalently, $C_m < m$), and 
        \item For $b > a$, $\Phi_a(y) \leq \Phi_b(y)$ (equivalently, $C_m \leq m$),
    \end{enumerate}
where $m = (a-b) \mod d$ and $C_m = |\{ i : t_i \geq 1-m/d\}|$). Indeed, if $\Phi_a(y) \geq \Phi_b(y)$ for any $b < a$, then $A(y) \neq a$ by our choice of tie-breaking rule in \cref{alg: inf_rounding}. Similarly, if $\Phi_a(y) > \Phi_b(y)$, then $A(y) \neq a$. The equivalent conditions $C_m < m$ and $C_m \leq m$ follow by unpacking the definition of $\Phi_b - \Phi_a$ in \cref{lem:phi_diff}.

Therefore, an exact characterization of $P_{a,z}$ is the set of points $y \in \mathbb{R}^d$ such that \eqref{eq:lem51} holds and $|\{ i : t_i \geq 1-m/d\}| < m$ for all $m \in [0, a]$ and $|\{ i : t_i \geq 1-m/d\}| \leq m$ for all $m \in [a+1, d-1]$. Equivalently\footnote{This equivalence is true up to boundary points, which will not change the final result.}, we may write the constraints on the coordinates $t_i$ as $0 \leq t_i < 1$ and once the $t_i$'s are sorted as $t_{(1)}, \dots, t_{(d)}$, then $t_{(d-m+1)} \leq 1-m/d$ for $m \leq a$ and $t_{(d-m)} \leq 1-m/d$ for $m > a$. Let $T$ be the set of points $(t_1, \dots, t_d)$ satisfying these constraints.  Now, consider.
    \[ \Delta_P(v) 
    = \inf_{y \in P} \max_{i \in [d]} |v_i - y_i|
    = \inf_{t \in T} \max_{i \in [d]} \left| v_i - \left( z_i + \frac{a}{d} + t_i \right) \right|
    = \inf_{t \in T} \max_{i \in [d]} \left| w_i - t_i \right|,\]
    where $w_i = v_i - (a/d) - z_i$.

    Let $\pi: [d] \rightarrow [d]$ be the permutation that orders $w_1, \dots, w_d$. We now show that we may assume $t_{\pi(1)} \leq \dots \leq t_{\pi(d)}$.  Indeed, if $t_i < t_j$ and $w_i > w_j$, we may safely swap $t_i$ and $t_j$ without increasing since $\max(|t_i - w_i|, |t_j - w_j|) \geq \max(|t_j - w_i|, |t_i - w_j|)$. 

    Therefore, we must minimize $r = \Delta_P(v)$ under the following constraints.
    \begin{enumerate}[nosep]
        \item $-r \leq w_i - t_i \leq r$ for all $i \in [d]$
        \item $0 \leq t_i < 1$ for all $i \in [d]$
        \item $t_{\pi(1)} \leq \dots \leq t_{\pi(d)}$
        \item $t_{\pi(d-m+1)} \leq 1-m/d$ for all $m \in [0, a]$
        \item $t_{\pi(d-m)} \leq 1-m/d$ for all $m \in [a+1, d-1]$
    \end{enumerate}
    This is a linear program with $\poly(d)$ variables and constraints, each of which can be represented using polynomially many bits, and can therefore be solved in time polynomial in the input length.
\end{proof}

Next, we show a bound on the number of
cells in a certain annulus.
\begin{lemma}
\label{lem:efficient_enum}
For every \(v\in\mathbb Q^d\), every \(m\ge0\), and $\tau=\rho/d$,
the number of cells \(P\in\cP\) with $ m\tau\le \Delta_P(v)\le (m+1)\tau$ is at most $N_m=(d+1)^{7(m+2)}$.
Moreover, these cells can be enumerated in time \(\poly(d) \cdot N_m\).
\end{lemma}

\begin{proof}
If $\Delta_P(v)\le (m+1)\tau,$
then \(P\) intersects the larger ball $\bbB := \calB_\infty(v; s\tau)$, where $s=m+2$.
Hence, it suffices to enumerate all cells that intersect $\bbB$, i.e., for each fixed $a$, to bound and enumerate
\[ W_a = \bbB \cap \left( \bigcup_{z \in \mathbb{Z}^d} P_{a,z} \right). \]
Define
\[
  F_i(a)=
  \left\{
  \left\lfloor \frac{y_i}{\rho}-\frac ad\right\rfloor:
  y\in \bbB
  \right\}.
\]
Since $\bbB$ has $\ell_\infty$-radius $s\tau=s\rho/d$, the interval of
possible values of $y_i/\rho$ has length $2s/d$.  We consider two cases depending on $s$. 

(i) Case $s\le d/3$.
Then \(2s/d<1\), so each \(F_i(a)\) has size at most \(2\).

Suppose $P_{a,z}\cap \bbB\neq\emptyset$, and choose \(y\in P_{a,z}\cap \bbB\).
For every coordinate $i$ with \(|F_i(a)|=2\), write $F_i(a)=\{r_i,r_i+1\}$.
Let
\[
  E(z)=\{i: |F_i(a)|=2 \text{ and } z_i=r_i\}.
\]
Thus \(E(z)\) is the set of coordinates where \(z\) chooses the lower of the two
possible floor values.

If \(i\in E(z)\), then \(y_i\) realizes the lower floor value, while some point
of \(\bbB\) realizes the upper floor value. Since the coordinate width of \(\bbB\) is
\(2s\tau\), this implies
\[
  \left\{\frac{y_i}{\rho}-\frac ad\right\}
  \ge
  1-\frac{2s}{d}.
\]
Because \(y\in P_{a,z}\), the chosen shift of \(y\) is \(a\). By \cref{lem:balancing}, $|E(z)|\le 2s+1$.

For fixed \(a\), the vector \(z\) is determined by \(E(z)\): coordinates with
\(|F_i(a)|=1\) have no choice, and coordinates with \(|F_i(a)|=2\) choose the
lower value exactly on \(E(z)\). Hence the number of candidates for this shift
is at most
\[
  \sum_{j=0}^{2s+1}\binom dj
  \le
  (d+1)^{2s+1}.
\]
Therefore, $|W_a| \leq (d+1)^{2s + 1}$. These candidates are enumerated by looping over \(a\) and all subsets
\(E\subseteq[d]\) of size at most \(2s+1\).

(ii) Case: \(s>d/3\).
Here we use the crude product enumeration. For fixed \(a\), each coordinate $y_i$ has
at most $\frac{2s}{d}+3$ possible floor values for $\lfloor y_i/\rho - a/d \rfloor$. Therefore, the number of candidate vectors \(\Round(y)\) for this
shift is at most
\[
  \left(\frac{2s}{d}+3\right)^d.
\]
Since for all $x \geq 1/3$, $\ln(2x + 3) \leq 4x$,
\[ \left(\frac{2s}{d}+3\right)^d \leq e^{4s} \leq (d+1)^{6s} .\]
The candidates are enumerated by looping over all $\leq (d+1)^{6s}$ possible combinations of coordinate floor values.

Combining the two cases and enumerating over all choices of $a$, the number of cells intersecting \(\bbB\) is at most
$d \cdot (d+1)^{6s} < (d+1)^{7(m+2)} = N_m.$  After enumerating these cells, we 
use \cref{lem:dist_comp} to compute \(\Delta_P(v)\) exactly in time
$\poly(d)$ and retain only  those cells satisfying $m\tau\le \Delta_P(v)\le (m+1)\tau$.
\end{proof}

\subsection{Sampling Algorithm}

The key technique we use is that of rejection sampling. Let $\tau = \lceil \frac{C}{\varepsilon}  \cdot \log(d+1) \rceil$ for a sufficiently large absolute constant $C$ and set $\rho = d \tau$. Choose $C$ so that $\varepsilon \tau \geq 400 B \cdot \log(d+1)$. Notice that we have chosen $\tau$ so that $\lambda^\tau \leq e^{-100 B \cdot \log(d+1)} = (d+1)^{-100 B}$.  We define several annuli of grid cells depending on their distance from $v$. These annuli partition the grid cells into disjoint sets:
\[ 
L_0 = \{ P : \lceil \Delta_P(v) \rceil \leq \tau \},
\qquad
L_m = \{ P: m \tau < \lceil \Delta_P(v)  \rceil \leq (m+1) \tau \}.
\]
Define $w_P := \lambda^{\lceil \Delta_P(v) \rceil}$.  We first enumerate $L_0$ and compute
\[ W_0 = \sum_{P \in L_0} w_P .\]
Note that since the cell containing $v$ has $\lceil \Delta_P(v) \rceil = 0$, $W_0 \geq 1$.
For $m \geq 1$, let
\[ U_m = N_m \lambda^{m \tau} =  (d+1)^{B(m+2)} \lambda^{m \tau},
\quad 
\gamma = \frac{U_{m+1}}{U_m} =  (d+1)^B \lambda^\tau \leq (d+1)^{-99B}.
\]
Finally, set
\begin{equation}
\label{eq:SU}
S = \sum_{m \geq 1} U_m = \frac{(d+1)^{3B} \lambda^\tau}{1-(d+1)^{B} \lambda^\tau}, \quad \mbox{ and } \quad 
U = W_0 + S.
\end{equation}

\begin{algorithm}[H]
\caption{Efficient Partition Sampler}
\label{alg:efficient_sampler}
\begin{algorithmic}
    \Require Dataset $X \in \cX^*$
    \State\hspace*{-5mm} \textbf{Parameters:} Function $f: \cX^* \to \R^d$, Privacy parameter $\varepsilon$, Representative function $\repr: \cP \to \R^d$
    \Statex
    \While{\textbf{true}}
        \State $b_0 \sim \text{Ber}(W_0/U)$ \Comment{Step 1: Possibly sample from $L_0$}
        \If{$b_0 = 1$}
            \State Sample a cell $P \in L_0$ with probability $w_P/W_0$ \Comment{Use \cref{lem:knuth-yao} repeatedly}
            \State \Return $r(P)$
        \EndIf
        \State
        \State $M \gets \mathrm{Geom}(1-\gamma)$ \Comment{Step 2: Sample $M$}
\State
        \State Enumerate $L_M = \{ P_1, \dots , P_h\}$ \Comment{Step 3: Possibly sample from $L_M$}
        \State Sample $J \sim \text{Uniform}([N_M])$
        \State
        \If{$J \le h$}
            \State $p \gets \lambda^{\lceil \Delta_{P_J}(f(X)) \rceil - M \tau}$
            \State $b_1 \sim \text{Ber}(p)$
            \If{$b_1 = 1$}
                \State \Return $r(P_J)$
            \EndIf
        \EndIf
    \EndWhile
\end{algorithmic}
\end{algorithm}

We use $\mathrm{Geom}(p)$ to denote the geometric distribution over $\Z_{> 0}$ where $k$ appears with probability $p (1 - p)^{k-1}$. This can be sampled from by repeatedly sampling from $\mathrm{Ber}(p)$ and return the index of the first sample that equals $1$, which takes $1/p$ samples from $\mathrm{Ber}(p)$ in expectation.

\begin{lemma}\label{lem:samples-correct}
\cref{alg:efficient_sampler} outputs $P \in \calP$ with probability proportional to $\exp(-\frac{\varepsilon}{2} \lceil \Delta_P(f(X)) \rceil)$.
\end{lemma}
\begin{proof}
Suppose $P \in L_0$. In one round, the probability of proposing $L_0$ and then choosing $P \in L_0$ is
\[ \frac{W_0}{U} \cdot \frac{w_P}{W_0} = \frac{w_P}{U} .\]
Now suppose $P \in L_m$, for some $m \geq 1$.  From \eqref{eq:SU}, $\gamma^{m-1} = U_m/U_1$, $S = U_1/(1-\gamma)$, and $\gamma^{m-1} (1-\gamma) = U_m/S$. In one round, the probability of proposing $L_m$, selecting $P$, and accepting is
\begin{align*}
& \frac{S}{U} \cdot \gamma^{m-1} (1-\gamma) \cdot \frac{1}{N_m} \cdot \lambda^{(\lceil \Delta_P(f(X)) \rceil - m \tau)}
= \frac{S}{U} \cdot \frac{U_m}{S} \cdot \frac{1}{N_m} \cdot \lambda^{(\lceil \Delta_P(f(X)) \rceil - m \tau)} \\
& = \frac{U_m}{U \cdot N_m} \cdot \lambda^{(\lceil \Delta_P(f(X)) \rceil - m \tau)} 
= \frac{\lambda^{m \tau}}{U} \cdot \lambda^{(\lceil \Delta_P(f(X)) \rceil - m \tau)} 
= \frac{\lambda^{\lceil \Delta_P(f(X))  \rceil}}{U} 
= \frac{w_P}{U}.
\end{align*}
    Therefore, the probability of outputting $P$ over all rounds is proportional to $w_P$.
\end{proof}

We prove a useful fact that for any sequence $\{C_t\}_{t \in \mathbb{N}}$ of random variables parameterized by and depending on round $t$ of \cref{alg:efficient_sampler}, $\mathbb{E}[C_1]$ is a fairly good estimate of the expectation of the sum of random variables. We then use this fact to bound the expected number of random bits consumed by our algorithm as well as the expected running time.
\begin{lemma}
    \label{lem:first_round_bound}
    Let $\{ \tilde{C}_t \}_{t \in \mathbb{N}}$ be a sequence of integer random variables where $\tilde{C}_t$ depends only on the randomness that would be used during round $t$ of \cref{alg:efficient_sampler} (if round $t$ were executed) and let $C_t = \tilde{C}_t$ if round $t$ of \cref{alg:efficient_sampler} executes and 0 otherwise.
    \[ \mathbb{E} \left[ \sum_{t \ge 1} C_t \right] = O(\mathbb{E}[C_1]). \]
\end{lemma}
\begin{proof}
    Let $\calE_t$ denote the event that round $t$ of the sampling algorithm results in an accepting computation (i.e., sampling happens and no restart occurs) and $p = \Pr[\calE_t]$ (note that the probability is the same across every round). Let $N = \min_t \one_{\calE_t}$; clearly,
    $\Pr[N \geq t] = (1-p)^{t-1}$.  Note that $C_t = \tilde{C}_t \cdot \one_{N \geq t}$ and $\mathbb{E}[\tilde{C}_t] = \mathbb{E}[C_1]$.
    \[ \mathbb{E} \left[ \sum_{t \ge 1} C_t \right]
    = \mathbb{E} \left[ \sum_{t \geq 1} \tilde{C}_t \cdot \one_{N \geq t} \right]
= \sum_{r \geq 1} \mathbb{E}[\tilde{C}_t] \cdot \Pr[N \geq t] 
    = \mathbb{E}[C_1] \sum_{t \geq 1} (1-p)^{t-1} = \frac{\mathbb{E}[C_1]}{p}.
    \]
Therefore, it suffices to show that $p = \Theta(1)$.  By \cref{lem:sampler_dist}, the probability in any round of 
    outputting $P$ is $w_P/U$, the acceptance probability in each
    round is
    \[ 
    p = \frac{\sum_{P \in \cP} \lambda^{\lceil \Delta_P(f(X)) \rceil}}{U}
    =
    \frac{\sum_{P \in \cP} \lambda^{\lceil \Delta_P(f(X)) \rceil}}{W_0 + S}
    \geq \frac{W_0}{W_0 + S} \geq \frac{1}{1+S} = \Omega(1),\]
    where the final inequality follows from \eqref{eq:SU} that  
    $\lim_{d \rightarrow \infty} S = 0$.  
\end{proof}

\subsection{Privacy}

We show that \cref{alg:efficient_sampler} samples from the correct distribution and is therefore $\varepsilon$-DP.

\begin{lemma}
\label{lem:sampler_dist}
\cref{alg:efficient_sampler} is $\varepsilon$-DP.
\end{lemma}
\begin{proof}
Since $f$ has sensitivity at most 1, by
\cref{lem:main-dp}, $\Delta_P(f(X))$
and hence $\lceil \Delta_P(f(X)) \rceil$ has sensitivity at most $1$. Thus, from \Cref{lem:samples-correct}, the $\eps$-DP property follows.
\end{proof}

\subsection{Utility}

The analysis of the total error is
similar to that of \cref{lem:utility}.
\begin{lemma}
    \label{lem:utility_v2}
    For $P$ sampled with probability proportional to $\exp(-\frac{\eps}{2} \lceil \Delta_P(f(X)) \rceil)$, it holds that $\expect[\Delta_P(f(X))] = O\!\left(\frac{\log d}{\varepsilon} \right)$; consequently, \cref{alg:efficient_sampler} has expected $\ell_\infty$-error $O\left(\frac{ d \log d}{\eps}\right)$.
\end{lemma}
\begin{proof}
This lemma is essentially the same as \cref{lem:utility}, we reiterate that argument for completeness. By \cref{thm: inf_main}, we have $k_d(\ell) \leq (\ell+1)^{O(d/\ell)} \leq e^{O(d \log d/\ell)}$. Then, by \cref{lem:tail_bound} applied with $\alpha = \log d$, for 
    some universal constants $c_1$ and
    $c_3 = 1/C$, we have
    \[ 
    \Pr[ P \notin \cP_{< (k-1)/\rho}]
        \leq 
    \exp\left(c_1 \log d- \log d \cdot d \frac{k-1}{C d \log d/ \varepsilon}\right) 
        = d^{c_1} \cdot \exp\left(-c_3 \varepsilon (k-1)\right). \]
    Let $k^* = \max\left(\lceil \frac{\rho}{d} \rceil + 1, \lceil \frac{c_1 \ln d}{c_3 \varepsilon} \rceil + 1\right)$. 
        Since $\frac{\rho}{d} = \frac{C \log d}{\varepsilon}$, we have $k^* = O\left(\frac{\log d}{\varepsilon}\right)$.  Now, 
        note that the choice of $k^*$ ensures that $e^{-c_3 \varepsilon (k^*-1)} \le 1/d^{c_1}$ and
        for $k > k^*$, we have $(k-1)/\rho \ge 1/d$.
\begin{align*}
    \expect[ \lceil \Delta_P(f(X)) \rceil]
    & = \sum_{k = 0}^\infty \Pr[\lceil \Delta_P(f(X)) \rceil > k] 
    \leq \sum_{k = 0}^\infty \Pr[\Delta_P(f(X)) > k-1] 
    \leq \sum_{k = 0}^\infty \Pr[ P \notin \cP_{< (k-1)/\rho}] \\
    & \leq \sum_{k=0}^{k^*} 1 + \sum_{k=k^*+1}^\infty d^{c_1} \cdot e^{-c_3 \varepsilon (k-1)} 
            \le k^* + 1 + \sum_{j=0}^\infty e^{-c_3 \varepsilon j}
    = k^* + 1 + \frac{1}{1 - e^{-c_3 \varepsilon}}\\
    & \leq O\left( \frac{\log d}{\varepsilon} \right) + O\left( \frac{1}{\varepsilon} \right)
    = O\left( \frac{\log d}{\varepsilon} \right).
    \end{align*}

    Since the expected $\ell_\infty$ error of \cref{alg:efficient_sampler} is at most $2 \rho + \Delta_P(f(X))$, substituting $\rho = C \cdot d \log d/\varepsilon$ and the above bounds lets us conclude that \cref{alg:efficient_sampler} has expected $\ell_\infty$ error at most $O(d \log d/\varepsilon)$.

\end{proof}

\subsection{Randomness Complexity}
\begin{lemma}
\label{lem:rand_complexity}
\cref{alg:efficient_sampler} has randomness complexity $O(\log d)$.
\end{lemma}
\begin{proof}
    Let $R_1$ denote the number of random bits used in round 1 of the sampling algorithm. By \cref{lem:first_round_bound}, the expected number of random bits used by the sampler is $O(\mathbb{E}[R_1])$. It therefore suffices to upper bound $\mathbb{E}[R_1]$.

    We write $R_1 = K_1 + K_2 + K_3$, where $K_1, K_2, K_3$ denote the number of random bits used in steps 1, 2, and 3 of the sampling algorithm respectively. First, $\mathbb{E}[K_1] = O(\log d)$ since $N_0 = \poly(d)$ and we can choose $P \in N_0$ by $O(\log d)$ calls to \cref{lem:knuth-yao}, using each call to decide which half of the remaining cells of $N_0$ to sample from. Next, $K_2$ is geometrically distributed with parameter $O(1-\gamma)$ and therefore, $\mathbb{E}[K_2] = O(\frac{1}{1-\gamma}) = O(1)$.
    
    Finally, $\mathbb{E}[K_3] = \underset{m \sim M}{\mathbb{E}}[\mathbb{E}[K_3 \mid M = m]]$. Since we require $\log(N_m)$ bits to sample from $[N_m]$ and $O(1)$ bits to sample $\text{Ber}(p)$, by \cref{lem:knuth-yao},
    \[ \mathbb{E}[K_3 \mid M =m] = \log(N_m) + O(1) = 7(m+2) \log(d+1) + O(1) = O(m \log d) . \]
    Therefore,
    \[ \mathbb{E}[K_3] 
    = O \left( \underset{m \sim M}{\mathbb{E}}[m \log d] \right) 
    = O \left( \log d \cdot \mathbb{E}[M] \right)
    = O \left(\log d \cdot \frac{1}{1-\gamma} \right) = O(\log d).\]
    So, $\mathbb{E}[R_1] = \mathbb{E}[K_1] + \mathbb{E}[K_2] + \mathbb{E}[K_3] = O(\log d)$, as desired.
\end{proof}

\subsection{Running Time}
\label{sec:time}

\begin{lemma}
    \label{lem:expected_time}
    Let $T(f, X)$ denote the time to evaluate $f$ on input $X$; consequently, \cref{alg:efficient_sampler} has expected running time $\poly(T(f, X), 1/\varepsilon)$.
\end{lemma}
\begin{proof}
    By \cref{lem:first_round_bound}, it suffices to bound $O(\mathbb{E}[C_1])$.  Let $|\cdot |_b$ denote the bit-length and $T := T(f, X)$.
    
    Let $K_1, K_2, K_3$ denote the amount of time spent in steps 1, 2, and 3 of \Cref{alg:efficient_sampler} respectively; we will bound these separately. 
    
    To bound $K_1$, note that we would spend $\poly(d)$ time in step 1 if all operations were over real numbers.  Furthermore, the computation of each $\lceil \Delta_P(f(X)) \rceil$ can be done in time polynomial in the length of $T$ and $d$ (which is $\poly(T)$) by \cref{lem:dist_comp}. Since these quantities are rational numbers with denominator $\leq (2^c)^{|\lceil \Delta_P(f(X)) \rceil|_b} \leq 2^{\polylog(1/\varepsilon) \cdot \poly(T)}$, we can perform bit operations with overhead bounded above by $\poly(T, 1/\varepsilon)$. Therefore, $\mathbb{E}[K_1] \leq \poly(T, 1/\varepsilon)$. 
    
    To bound $K_2$, we observe that we may write $K_2 = \sum_{j=1}^\infty R_j$, where $R_j$ denotes the bit complexity of round $j$ of step 2 (where $R_j = 0$ if the process has terminated). 
    \[ \mathbb{E}[R_j] = \mathbb{E}[R_j \mid R_j > 0] \cdot \Pr[R_j > 0] = \mathbb{E}[R_j \mid R_j > 0] \cdot \gamma^{j-1} = \poly(d, 1/\varepsilon, \log j) \cdot \gamma^{j-1}.\] From this, 
    \begin{align*}
        \mathbb{E}[K_2] 
        &= \sum_{j=1}^\infty \mathbb{E}[R_j] 
        = \poly(d, 1/\varepsilon)  \sum_{j=1}^\infty \polylog(j) \cdot \gamma^{j-1}
        = \poly(d, 1/\varepsilon).
    \end{align*}

    To bound $\mathbb{E}[K_3]$, we compute $\mathbb{E}[K_3 \mid M = m]$.  By \Cref{lem:efficient_enum}, $L_m$ can be enumerated in time $\poly(d) \cdot N_m = \poly(d) \cdot (d+1)^{B (m+2)}$ and choosing $J$ can be done in (expected) time $\poly(\log N_m) = \poly(m, \log d)$.  Next, consider the bit complexity of $\lceil \Delta_{P_j}(f(X)) \rceil$.  This quantity is computed by a polynomial time call to the subroutine \cref{lem:dist_comp} with input $f(X)$ and $a, z$ representing $P_j$; this yields $|f(X)|_{b} \leq T$ and $|a, z|_b \leq \poly(\tau m) = \poly(1/\varepsilon, \log d, m)$. Hence, $\lceil \Delta_{P_j} (f(X)) \rceil \leq \poly(T, 1/\varepsilon, m)$ and $|\lambda^{\lceil \Delta_{P_j}(f(X)) \rceil}|_b = \poly(T, 1/\varepsilon, m)$. Therefore, $p$ can be sampled in expected time $\poly(T, 1/\varepsilon, m)$.  Together, $\mathbb{E}[K_3 \mid M=m] = \poly(T, 1/\varepsilon, m) \cdot (d+1)^{B(m+2)}$.
    \begin{align*}
        \mathbb{E}[K_3] 
        &= \sum_{m=1}^\infty \mathbb{E}[K_3 \mid M=m] \cdot \Pr[M=m] \\
        &= \sum_{m=1}^\infty \poly(T, 1/\varepsilon, m) \cdot (d+1)^{B(m+1)}  \cdot \gamma^{m-1} (1-\gamma)\\
        &\leq \sum_{m=1}^\infty \poly(T, 1/\varepsilon, m) \cdot (d+1)^{B(m+2)} \cdot \gamma^{m-1}\\
        &\leq \sum_{m=1}^\infty \poly(T, 1/\varepsilon, m) \cdot (d+1)^{B(m+2)} \cdot (d+1)^{B (m-1)} (d+1)^{-100B (m-1)}\\
        &\leq \poly(T, 1/\varepsilon) + \sum_{m=1}^\infty \poly(T, 1/\varepsilon, m) \cdot (d+1)^{-m} 
        = \poly(T, 1/\varepsilon).  \qedhere
    \end{align*}
\end{proof}

 \section{Conclusion and Open Questions} \label{sec:open}

We give an $\eps$-DP algorithm for answering $d$ linear queries with nearly optimal error and randomness complexity, which (essentially) resolves the main open questions from \cite{canonne2025randomness}. Despite this progress, there are several obvious remaining gaps, which we highlight below.
\begin{itemize}
\item As mentioned earlier, our randomness complexity of $O(\log d)$ is only known to be tight in the regime where $\eps \leq 1/d$. Moreover, our efficient mechanism has $\ell_\infty$-error $O(d \log d / \eps)$ which is a factor of $O(\log d)$ off from the optimal bound. \item Our algorithms are only efficient for the $\ell_\infty$-sensitivity case. It would be interesting to obtain efficient algorithms for all $K$-norms, or at least all $\ell_p$-norms.
\item While not the focus of our paper, we do obtain improvements for \emph{approximate}-DP (i.e., $(\eps, \delta)$-DP) over the previous algorithms \cite{canonne2025randomness,ghentiyala2026efficient} as well. Without a constraint on randomness complexity, the optimal $\ell_\infty$-error is $O(\sqrt{d \log(1/\delta)} / \eps)$~\cite{SteinkeU16,DaganK22,GhaziKM21}. If we restrict the {\em worst case} number of random bits to $(\log d)^{O(1)}$, then the \emph{approximate}-DP mechanisms from \cite{canonne2025randomness,ghentiyala2026efficient} incur $\ell_\infty$-error of at least $\tilde{\Omega}(d^{1.5} / \eps)$, which is \emph{worse} than our \emph{pure}-DP algorithm with $O(d / \eps)$ error. Nevertheless, there is still a gap of $\tilde{\Omega}(\sqrt{d})$ between our error bound and the upper bound without any worst-case randomness complexity constraint.
\end{itemize}

\noindent Finally, we remark that the study of DP algorithms under the lens of randomness complexity is still nascent, and a large number of problems are yet to be explored. Specifically, one could take any fundamental DP problem and examine its randomness complexity vs utility tradeoff. 
\section*{AI Disclosure Statement}

Gemini 3.1 Pro and GPT 5.5 Pro were used for several parts of this manuscript. The authors defined a multi-scale secluded-partition (MSSP) and observed that such objects may be used to obtain particularly low randomness DP mechanisms. AI was used to construct a good MSSP for the $\ell_\infty$-norm, which it did by modifying the construction of Hoza and Klivans \cite{HozaK18}. The authors confirmed the correctness of this construction by formalizing it in {\tt Lean4} using an AI coding agent (Google Antigravity). The authors then observed that an efficient mechanism could be obtained using the ring sampling technique described in \cref{sec:efficient}, which AI was used to formalize. AI was also used to generate the Rogers \cite{Rogers_1957} inspired construction in \cref{subsec:k-norm}. AI was used to generate drafts of almost all proofs, which the authors then iterated over. AI was also used to generate \cref{fig:secluded_partition} and check proofs for errors. The authors verified the correctness and originality of all content and are ultimately responsible for any oversight. 
\bibliographystyle{alpha}
\bibliography{main}

\end{document}